%% file: main_.tex
\documentclass[accepted]{uai2026} 
                        
\usepackage[american]{babel}

\usepackage{natbib} 
\usepackage{mathtools} 
\usepackage{booktabs} 
\usepackage{tikz} 
\usetikzlibrary{positioning,calc}
\usepackage{amsthm}
\usepackage{comment}
\usepackage{thmtools}
\usepackage{amsmath}
\usepackage{cleveref}
\usepackage{mathbbol}
\usepackage{cancel}
\usepackage[normalem]{ulem}
\usepackage{subcaption}
\declaretheorem[name=Theorem,within=section]{thm}

\declaretheorem[name=Assumption,numberlike=thm]{asn}

\DeclareMathOperator{\expit}{expit}

\DeclareMathOperator{\pa}{pa}

\def\ci{\perp\!\!\!\perp}

\usepackage[normalem]{ulem}
\newcommand{\stkout}[1]{\ifmmode\text{\sout{\ensuremath{#1}}}\else\sout{#1}\fi}

\usepackage{enumitem}
\setlist{nolistsep}

\title{Proximal Identification and Estimation in Front-Door Causal Structures with Unobserved Confounding of the Mediator}

\author[1]{\href{mailto:<hguo51@jh.edu>?Subject=Proximal Confounded Mediator Paper}{Helen Guo}{*}}
\author[2]{Beatrix Yaxin Wen}
\author[3]{Ilya Shpitser}
\affil[1]{%
    Biostatistics Dept.\\
    Johns Hopkins Bloomberg School of Public Health\\
    Baltimore, MD, USA
}
\affil[2]{%
    Applied Mathematics and Statistics Dept.\\
    Johns Hopkins Whiting School of Engineering\\
    Baltimore, MD, USA
}
\affil[3]{%
    Computer Science Dept.\\
    Johns Hopkins Whiting School of Engineering\\
    Baltimore, MD, USA
  }
  
\begin{document}
\maketitle
\begin{abstract}
Unobserved confounding is a fundamental obstacle in causal inference problems. In the graphical modeling literature, a general theory has been developed that allows identification in the presence of hidden variables, with some limitations. In particular, Pearl's celebrated front-door criterion allows nonparametric identification in the presence of unobserved common causes of the treatment and the outcome, however it requires the presence of an unconfounded variable that mediates all causal influence from the treatment to the outcome. This stringent requirement limits the applicability of the front-door criterion.
We propose proximal generalizations of the front-door criterion, allowing both arbitrary treatment/outcome confounding, and unobserved confounders of the mediator, provided informative proxies for the latter type of confounders are observed.  In addition to deriving three new identification strategies in this setting, we provide plug-in and influence function-based estimation strategies for the resulting functionals, and evaluate their performance through simulations.
\end{abstract}

\input{1Introduction}
\input{Background}

\input{Identification}

\input{K_Estimation}

\input{L_Experiments}
\input{N_Discussion}

\newpage

\begin{contributions}
Helen Guo, Beatrix Wen, and Ilya Shpitser developed the identification results.
Helen Guo developed the semiparametric estimation results.
Beatrix Wen and Helen Guo designed and conducted the simulation studies.
\end{contributions}

\begin{acknowledgements}
This work was supported by the Office of Naval Research (Grant No.~N000142412701), the National Institutes of Health (Grant No.~1R56AI191526-01), the National Science Foundation CAREER Award (Grant No.~1942239), and the National Library of Medicine, (Grant No. 1R01LM014800-01A1).
\end{acknowledgements}

\bibliographystyle{plainnat}
\bibliography{refs}

\newpage

\onecolumn

\title{Supplementary Material}
\maketitle



\appendix

\input{S_Appendix}

\end{document}

%% file: 1Introduction.tex
\section{introduction}

Causal effects are not 
identified in general in the presence of hidden common causes of treatment and outcome variables. A rich body of work has proposed strategies to overcome this obstacle via nonparametric identification theory developed in the graphical modeling literature, leading to the ID algorithm and extensions \citep{pearl1995causal,tianIdentificationCausalEffects2002,shpitser06id, huang06do}, or via methods that rely on fortuitous external variables such as instruments, mediators, or proxies (i.e., causal consequences of hidden variables) \citep{newey03instrumental, kuroki2014measurement, miao2018identifying}. Recent work \citep{shpitser2023proximal} has combined graphical model-based identification theory and proximal methods to obtain identification in settings where neither approach alone suffices.



We present proximal causal inference results in a novel setting, extending Pearl's front-door method \citep{pearl1995causal} to settings where unobserved common causes of the mediator, treatment, and outcome exist.
Specifically, we consider a causal model represented by the \emph{composite bow graph} in Fig.~\ref{fig:primal_v0}. This model captures a common scenario in which some hidden variable(s) confound only treatment \(A\) and outcome \(Y\), while others confound treatment \(A\), mediator \(\vec{M}\), and outcome \(Y\). Such structures arise frequently in practice, as mediators often share unmeasured causes with both treatment and outcome, thereby violating the assumptions required for identification based on the front-door criterion.

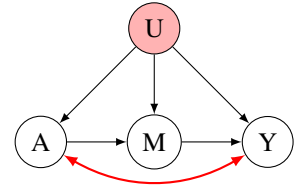
\begin{figure}[h!]
\centering
\begin{tikzpicture}[>=latex, node distance=1.5cm, anchor=center]

\begin{scope}
    \node (A) [draw, circle] {A};
    \node (M) [draw, circle, right of=A] {M};
    \node (Y) [draw, circle, right of=M] {Y};
    \node (U) [draw, circle, fill=red!30, above of=A, xshift=1.5cm] {U};

    \draw[->] (A) -- (M);
    \draw[->] (M) -- (Y);
    \draw[->] (U) -- (Y);
    \draw[->] (U) -- (A);
    \draw[->] (U) -- (M);
    \draw[<->, red, thick, bend right=30] (A) to (Y);

\end{scope}
\end{tikzpicture}
\caption{The composite bow graph.}
\label{fig:primal_v0}
\end{figure}

We show that when the latent cause(s) affecting \(A\), \(\vec{M}\), and \(Y\) in the composite bow graph admit proxies, the counterfactual distribution \(p(Y(a))\) (where $Y(a)$ denotes the \emph{potential outcome} had the treatment variable $A$ taken value $a$) and thus the causal effect of \(A\) on \(Y\) remains identifiable under mild structural assumptions. We provide three distinct identification strategies, each relying on a different set of structural assumptions. In addition, we develop estimators for functionals obtained by these strategies, including plug-in estimators and an influence function-based estimator.
Finally, we illustrate the validity of our identification methods 
and performance of our proposed estimators through simulation studies.

This paper is organized as follows. Section~\ref{sec:background} reviews relevant background and outlines our contributions. Section~\ref{sec:identification} presents distinct 
identification strategies under different structural assumptions. Section~\ref{sec:estimation} develops estimators of the causal effect under these strategies. Section~\ref{sec:experiments} evaluates the performance of the proposed estimators via simulation, and Section~\ref{sec:discussion} concludes this work with a discussion.

%% file: Background.tex
\section{Background}\label{sec:background}
\subsection{
Graphical Causal Models
}
We represent fully observed causal structures by means of models associated with directed acyclic graphs (DAG)s. In such graphs,
directed edges $(\rightarrow)$ encode direct causal relationships between variables. Conditional independence relations implied by the model are given by the graph via \emph{d-separation}. Such causal structures are able to represent responses of outcome variables $\vec{Y}$ where treatment variables $\vec{A}$ are set to values $\vec{a}$ by means of an external intervention operation denoted by the do operator $\text{do}(\vec{a})$ \citep{pearl2009causality}; these distributions are also written in the \emph{potential outcome} notation as $p(\vec{Y}(\vec{a}))$.
Identification of the resulting interventional distributions 
in causal models of a DAG is obtained by means of the \emph{g-formula}:
\begin{align*}
p(\vec{Y} \mid \text{do}(\vec{a}))
= \sum_{\vec{Y}^* \setminus \vec{Y}}
\left( \prod_{V \in \vec{Y}^*}
p(V \mid \pa_{\cal G}(V)) \middle) \right|_{\vec{A} = \vec{a}}
\end{align*}
Identification of interventional distributions leads to natural estimators of causal parameters of interest, such as the \emph{average causal effect (ACE)}: $\mathbb{E}[Y(A=1) - Y(A=0)]$ for binary $A$.

Causal structures with hidden variables 
are
represented using \emph{acyclic directed mixed graphs} (ADMGs) containing both directed ($\to$) and bidirected ($\leftrightarrow$) edges, and obtained from a hidden variable DAG by means of the latent projection operation \citep{verma90equiv}. In this paper we will also represent hidden variable causal models by \emph{partial latent projections}, where some hidden variables are ``projected out'' from the graph (replaced by bidirected edges) and other hidden variables are explicitly retained (shown in red).  The composite bow graph in Fig.~\ref{fig:primal_v0} is an example of a partial latent projection. Conditional independence relations in these graphs are given by \emph{m-separation criterion} \citep{richardson02ancestral}, the natural extension of d-separation to ADMGs. 

In hidden variable models, {the task of identification of}
interventional distributions becomes considerably more complicated.  Some interventional distributions are not identified at all, while others are expressible as a functional of the observed data given by the ID algorithm \citep{tianIdentificationCausalEffects2002,shpitser06id, huang06do}. A number of criteria have been developed to classify DAGs or ADMGs where the interventional distribution is identified by a particular simple functional. We now describe two such examples.




\subsection{The back-door Criterion}
The back-door criterion, introduced by \citet{pearl93comment}, provides graphical conditions under which interventional distribution \(p(Y(a))\) and thus the causal effect of \(A\) on \(Y\)
is identified for binary $A$, along with an identifying formula.

A set of observed variables $\vec{C}$ satisfies the back-door criterion relative to an ordered treatment-outcome pair ($A,Y$) in 
an ADMG ${\cal G}$ if: 
\begin{itemize}
    \item[1.] No node in $\vec{C}$ is a descendant of $A$; and 
    \item[2.] $\vec{C}$ m-separates every back-door path between $A$ and $Y$ (e.g. every path from $A$ to $Y$ that starts with an arrowhead into
    $A$).
\end{itemize}
For example, \(\vec{C}\) satisfies the back-door criterion relative to $A$ and $Y$ in the simple case shown in Fig. \ref{fig:backdoor_obs}, where $\vec{C}$ is a common cause of $A$ and $Y$. If \(\vec{C}\) satisfies the back-door criterion, the interventional distribution \(p(Y(a))\) is identified from the observed data distribution $p(Y,A,\vec{C})$ by the \emph{adjustment formula}, shown below in Equation~\ref{eq:back-door}:
\begin{equation}\label{eq:back-door}
p(Y(a)=y) = \sum_{\vec{c}}p(y|a,\vec{c})p(\vec{c})
\end{equation}
Had a common cause of $A$ and $Y$ not been observed, shown as a red vertex $U$ in Fig.~\ref{fig:backdoor_hidden},
the 
interventional distribution \(p(Y(a))\) 
is generally not identified from the observed data distribution, which is now $p(Y,A)$.






\begin{figure}[htbp]
\centering

\begin{subfigure}{0.1\textwidth}
\centering
\begin{tikzpicture}[>=latex, node distance=1cm]
    \node (A) [draw, circle] {A};
    \node (Y) [draw, circle, right of=A] {Y};
    \node (C) [draw, circle, above of=A, xshift=0.5cm] {C};

    \draw[->] (C) -- (Y);
    \draw[->] (C) -- (A);
    \draw[->] (A) -- (Y);
\end{tikzpicture}
\caption{}
\label{fig:backdoor_obs}
\end{subfigure}
\hspace{0.01\textwidth}
\begin{subfigure}{0.1\textwidth}
\centering
\begin{tikzpicture}[>=latex, node distance=1cm]
    \node (A) [draw, circle] {A};
    \node (Y) [draw, circle, right of=A] {Y};
    \node (U) [draw, circle, fill=red!30, above of=A, xshift=0.5cm] {U};

    \draw[->] (U) -- (Y);
    \draw[->] (U) -- (A);
    \draw[->] (A) -- (Y);
\end{tikzpicture}
\caption{}
\label{fig:backdoor_hidden}
\end{subfigure}
\hspace{0.01\textwidth}
\begin{subfigure}{0.2\textwidth}
\centering
\begin{tikzpicture}[>=latex, node distance=1cm]
    \node (A) [draw, circle] {A};
    \node (M) [draw, circle, right of=A] {M};
    \node (Y) [draw, circle, right of=M] {Y};

    \draw[->] (A) -- (M);
    \draw[->] (M) -- (Y);
    \draw[<->, red, thick, bend left=60] (A) to (Y);
\end{tikzpicture}
\caption{}
\label{fig:backdoor_mediation}
\end{subfigure}

\caption{(a) Observed confounding; (b) Hidden confounding; (c) Observed mediation with hidden treatment-outcome confounding.}
\label{fig:back-door}
\end{figure}
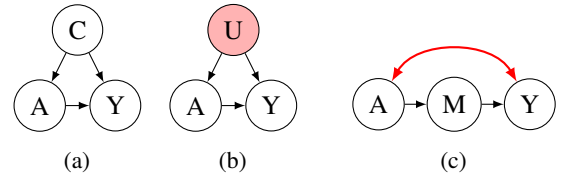

A recent line of work addressed lack of identifiability in the presence of a hidden common cause $U$ in Fig.~\ref{fig:backdoor_hidden} by taking advantage of informative proxy variables for $U$.  In particular, methods developed by \citet{miao2018identifying}, \citet{kuroki2014measurement}, and \citet{allman2015parameter} show that when proxies for an unobserved confounder are available (e.g., $W$ and $Z$ in Fig.~\ref{fig:miao}), along with mild structural assumptions, the counterfactual distribution $p(Y(a))$ can be identified. In particular, the adjustment formula,
\(
\sum_{\vec{u}} p(Y \mid a, \vec{u})\,p(\vec{u}),
\)
can be re-expressed in terms of the observed data distribution, even though $U$ is not observed directly.

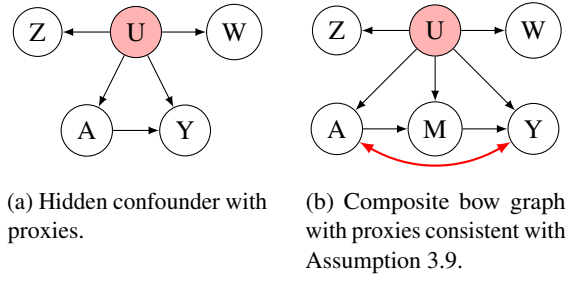
\begin{figure}[ht]
\centering

\begin{subfigure}[t]{0.2\textwidth}
\centering
\begin{tikzpicture}[>=latex, node distance=1.3cm, anchor=center,
                    transform shape]
    \node (A) [draw, circle] {A};
    \node (U) [draw, circle, fill=red!30, above of=A, xshift=.65cm] {U};
    \node (Y) [draw, circle, below of=U, xshift=.65cm] {Y};
    \node (W) [draw, circle, right of=U] {W};
    \node (Z) [draw, circle, left of=U] {Z};

    \draw[->] (A) -- (Y);
    \draw[->] (U) -- (Z);
    \draw[->] (U) -- (W);
    \draw[->] (U) -- (Y);
    \draw[->] (U) -- (A);
    \draw[<->, white, thick, bend right=52] (A) to (Y);
\end{tikzpicture}
\caption{Hidden confounder with proxies.}
\label{fig:miao}
\end{subfigure}
\hspace{0.02\textwidth}
\begin{subfigure}[t]{0.2\textwidth}
\centering
\begin{tikzpicture}[>=latex, node distance=1.3cm, anchor=center,
                    transform shape]
    \node (A) [draw, circle] {A};
    \node (M) [draw, circle, right of=A] {M};
    \node (Y) [draw, circle, right of=M] {Y};
    \node (U) [draw, circle, fill=red!30, above of=A, xshift=1.3cm] {U};
    \node (W) [draw, circle, right of=U] {W};
    \node (Z) [draw, circle, left of=U] {Z};

    \draw[->] (A) -- (M);
    \draw[->] (M) -- (Y);
    \draw[->] (U) -- (Z);
    \draw[->] (U) -- (W);
    \draw[->] (U) -- (Y);
    \draw[->] (U) -- (A);
    \draw[->] (U) -- (M);
    \draw[<->, red, thick, bend right=30] (A) to (Y);
\end{tikzpicture}
\caption{Composite bow graph with proxies consistent with Assumption~\ref{assump:ci_kp}.}
\label{fig:primal_v3}
\end{subfigure}

\caption{Causal graphs with proxy variables.}
\label{fig:combined_1}
\end{figure}

\subsection{The front-door Criterion}
Leveraging mediator(s) \(\vec{M}\) between \(A\) and \(Y\), the front-door criterion introduced in \citet{pearl1995causal}, provides alternative graphical conditions under which interventional distribution \(p(Y(a))\) is identified, while allowing for hidden common causes of treatment \(A\) and outcome \(Y\).

A set of observed variables $\vec{M}$ satisfies the front-door criterion relative to 
an ordered pair ($A$, $Y$) in 
an ADMG ${\cal G}$ if:
\begin{enumerate}
    \item $\vec{M}$ intercepts all directed paths from $A$ to $Y$;
    \item Every back-door path (paths starting with an arrowhead into $A$) between $A$ and $\vec{M}$ is marginally m-separated;
    \item $A$ m-separates all back-door paths (paths starting with an arrowhead into $\vec{M}$) between $\vec{M}$ and $Y$.
\end{enumerate}

 If \(\vec{M}\) satisfies the front-door criterion, the interventional distribution \(p(Y(a))\) is identified via the \emph{front-door 
 formula}:
\begin{equation}\label{eq:front-door}
p(Y(a)=y) = \sum_{\vec{m},a'} p(y\mid a',\vec{m})p(\vec{m}\mid a)p(a')
\end{equation}
In the ADMG in Fig. \ref{fig:backdoor_mediation}, where the bidirected arc represents a hidden common cause of $A$ and $Y$, \(\vec{M}\) satisfies the front-door criterion relative to $(A,Y)$.

\subsection{
Outline of Our Contributions
}

A common criticism of the front-door criterion in practice is that the assumptions regarding the mediator variable(s) are rarely met in practical problems where unobserved confounding is a concern.  Specifically, while the front-door criterion allows unobserved common causes of $A$ and $Y$, it specifically precludes such causes from influencing the mediator(s) $\vec{M}$.

A more natural generalization of the front-door graph, which we call the 
\emph{composite bow graph,} is shown in Fig. \ref{fig:primal_v0}.
In this graph,
\(\vec{M}\) does not satisfy the front-door criterion due to the presence of a hidden common cause \(\vec{U}\) of \(A\), \(\vec{M}\), and \(Y\). Indeed, $p(Y(a))$ is generally not identifiable in  Fig. \ref{fig:primal_v0}. 

We first note that
had $\vec{U}$ been observed rather than hidden in Fig.~\ref{fig:primal_v0}, the interventional distribution \(p(Y(a))\) 
would have been
identified by a mild generalization of the front-door formula given by Equation~\ref{eq:front-door2} \citep{fulcher2020robust}:
\begin{equation}\label{eq:front-door2}
p(Y(a)=y) =\sum_{\vec{m},a',\vec{u}}p(y\mid a',\vec{m},\vec{u})p(\vec{m}\mid a,\vec{u})p(a',\vec{u}) 
\end{equation}
Next, following the work in \citet{miao2018identifying}, \citet{kuroki2014measurement}, \citet{allman2015parameter}, and \citet{dukes2023proximal},
we show that with the existence of proxies \(W\) and \(Z\) for \(\vec{U}\), along with mild structural assumptions, $p(Y(a))$ may be identified via proximal generalizations 
of the formula in Equation~\ref{eq:front-door2}. We present three different sets of assumption that each guarantee identification in Section~\ref{sec:identification}. 



Our first two identification strategies (Sections~\ref{subsec:id_set1} and~\ref{subsec:id_set2}) adapt the assumptions similar in spirit to those in \citet{miao2018identifying}. These approaches rely on distributional variation restrictions together with the existence of solutions to certain Fredholm integral equations of the first kind.
This line of work has been extended by \citet{ghassami2025causal} to identify causal effects leveraging proxies of hidden mediator(s), and by \citet{dukes2023proximal} to identify direct and indirect effects in the presence of hidden confounding. The first identification strategy we present has previously appeared in \citet{bai2025}, which gives identification and estimation results for the population indirect effect, a functional identical to that of the causal effect in our model.

Our third identification strategy (Section~\ref{subsec:id_set3}) is more closely related to a result from \citet{kuroki2014measurement}, and \citet{allman2015parameter}. 
This strategy imposes a different set of distributional variation restrictions and enables identification of the full joint distribution, from which $p(Y(a))$ 
may be identified via the g-formula.
This line of work traces back to work on unique tensor (multidimensional array) decompositions, 
originating from \citet{kruskal1977three}.

%% file: Identification.tex
\section{Identification}\label{sec:identification}

We now describe the three sets of assumptions we use to allow identification of the interventional distribution $p(Y(a))$ in the composite bow graph when latent cause(s) affecting \(A, \vec{M},\) and \(Y\) admit two proxies. Throughout this work, whenever we consider causal effects, the treatment variable \(A\) is assumed to be binary, taking values in \(\{0,1\}\).

\subsection{Assumption Set 1}\label{subsec:id_set1}

The first identification strategy may be viewed as a front-door analogue of the proximal result in \citet{dukes2023proximal} for direct and indirect effects, and has previously appeared in \citet{bai2025}, which established identification and semiparametric estimation results for population indirect effects. The strategy combines conditional independence relations reflected in Fig.~\ref{fig:primal_v1} (Assumption~\ref{assump:ci_var1}) along with the structural conditions in Assumptions~\ref{assump:completeness_var1}--\ref{assump:bridge_var1}. 

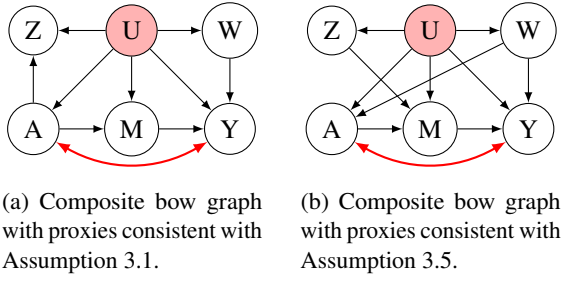
\begin{figure}[ht]
\centering

\begin{subfigure}[t]{0.2\textwidth}
\centering
\begin{tikzpicture}[>=latex, node distance=1.3cm, anchor=center,
                    transform shape]
    \node (A) [draw, circle] {A};
    \node (M) [draw, circle, right of=A] {M};
    \node (Y) [draw, circle, right of=M] {Y};
    \node (U) [draw, circle, fill=red!30, above of=A, xshift=1.3cm] {U};
    \node (W) [draw, circle, right of=U] {W};
    \node (Z) [draw, circle, left of=U] {Z};

    \draw[->] (A) -- (M);
    \draw[->] (M) -- (Y);
    \draw[->] (U) -- (Z);
    \draw[->] (U) -- (W);
    \draw[->] (U) -- (Y);
    \draw[->] (U) -- (A);
    \draw[->] (U) -- (M);
    \draw[->] (A) -- (Z);
    \draw[<->, red, thick, bend right=30] (A) to (Y);
    \draw[->] (W) -- (Y);
\end{tikzpicture}
\caption{Composite bow graph with proxies consistent with Assumption \ref{assump:ci_var1}.}
\label{fig:primal_v1}
\end{subfigure}
\hspace{0.02\textwidth}
\begin{subfigure}[t]{0.2\textwidth}
\centering
\begin{tikzpicture}[>=latex, node distance=1.3cm, anchor=center,
                    transform shape]
    \node (A) [draw, circle] {A};
    \node (M) [draw, circle, right of=A] {M};
    \node (Y) [draw, circle, right of=M] {Y};
    \node (U) [draw, circle, fill=red!30, above of=A, xshift=1.3cm] {U};
    \node (W) [draw, circle, right of=U] {W};
    \node (Z) [draw, circle, left of=U] {Z};

    \draw[->] (A) -- (M);
    \draw[->] (M) -- (Y);
    \draw[->] (U) -- (Z);
    \draw[->] (U) -- (W);
    \draw[->] (U) -- (Y);
    \draw[->] (U) -- (A);
    \draw[->] (U) -- (M);
    \draw[->] (W) -- (Y);
    \draw[->] (W) -- (A);
    \draw[->] (Z) -- (M);
    \draw[<->, red, thick, bend right=30] (A) to (Y);
\end{tikzpicture}
\caption{Composite bow graph with proxies consistent with Assumption \ref{assump:ci_var2}.}
\label{fig:primal_v2}
\end{subfigure}

\caption{Composite bow graphs under different proxy assumptions.}
\label{fig:combined_2}
\end{figure}




\begin{asn}\label{assump:ci_var1}
(Conditional Independences)
\begin{itemize}[noitemsep]
\item[1.] $Y \ci Z \mid \vec{U}, \vec{M}, A$
\item[2.] $W \ci A, \vec{M}, Z \mid \vec{U}$
\item[3.] $W, \vec{M} \ci Z \mid \vec{U}, A$
\end{itemize}
\end{asn}
Assumption~\ref{assump:ci_var1} 
may be represented graphically via the m-separation criterion in the variant of the composite bow graph shown in Fig.~\ref{fig:primal_v1}.  Note that this assumption allows association between the proxy $Z$ and the treatment $A$, as well as the proxy $W$ and the outcome $Y$.



\begin{asn}\label{assump:completeness_var1}
(Completeness) For each \((a, \vec{m})\) and any square integrable \(g\),
    \begin{itemize}[noitemsep]
        \item[1.] $\mathbb{E}\{g(\vec{u}) \mid z, a,\vec{m}\}=0$ a.s. iff $g(\vec{u})=0$ a.s. 
        \item[2.] $\mathbb{E}\{g(\vec{u}) \mid z,a\}=0$ a.s. iff $g(\vec{u})=0$ a.s. 
    \end{itemize}
\end{asn}
Assumption~\ref{assump:completeness_var1} imposes distributional variation restrictions referred to as \emph{completeness} assumptions and can be satisfied without parametric distributional restrictions \citep{newey03instrumental, Hu_Completeness}. {The assumption states that objects in $\{ p(\vec{u}_i \mid Z, a, \vec{m}) : \vec{u}_i \}$ exhibit sufficient variation across values $\vec{u}_i$ for each $(a,\vec{m})$.  Applying Bayes' rule, this assumption also implies sufficient variation of distributions in $\{ p(Z \mid \vec{u}_i,a,\vec{m}) : \vec{u}_i \}$.  
Note that these latter distributions are not functions of $\vec{m}$ by Assumption~\ref{assump:ci_var1}.} Additionally, if the dimension of \(Z\) is greater than that of \(U\), distributions that fail to satisfy Assumption~\ref{assump:completeness_var1} can be made arbitrarily close in total variation distance to distributions that do satisfy it \citep{cui2024semiparametric, Canay_testability2013}. Since completeness assumptions are not testable, the plausibility of Assumption~\ref{assump:completeness_var1} relies on domain knowledge. For a discussion of completeness assumptions sufficient to understand this paper, see \citet{Guo2025}.

To build intuition for the meaning of this assumption, consider a finite-support setting in which we form a stochastic matrix which has rows \(i\) corresponding to $p(\vec{u}_i \mid Z,a,\vec{m})$. This matrix may be viewed as defining a linear operator that maps a function $g(\vec U)$ to its conditional expectation given $Z$. 
For example, multiplying the matrix by a vector with entries $g(\vec{u}_i)$ yields
\begin{equation}\label{eq:completeness_ex1}
\sum_{\vec{u}_i} g(\vec{u}_i)\, p(\vec{u}_i \mid Z,a,\vec{m})
= \mathbb{E}\!\left[ g(\vec U) \mid Z,a,\vec{m} \right]
\end{equation}
If the rows of the matrix are linearly independent, then the matrix has full row rank. Assumption~\ref{assump:completeness_var1} imposes a continuous analogue of this full row rank requirement, so that if there exists a function $h_1$ such that
\begin{equation}
\begin{aligned}
&\sum_{\vec{u}}p(y \mid \vec{u},a,\vec{m},z)p(\vec{u} \mid z, a,\vec{m})\\
&= \sum_{w, \vec{u}} h_1(\cdot)p(w \mid u,a,\vec{m},z)p(\vec{u} \mid z, a,\vec{m}),
\end{aligned}
\end{equation}
then 
\begin{equation}
p(y \mid \vec{u},a,\vec{m},\cancel{z}) = \sum_{w} h_1(\cdot)p(w \mid \vec{u},\cancel{a,\vec{m},z}),
\end{equation}
where independence of the left hand side of $z$ and right hand side of $a,\vec{m},z$ follow by Assumption~\ref{assump:ci_var1}.



Completeness enables identification in tandem with the following assumption.

\begin{asn}\label{assump:bridge_var1}
(Fredholm Equation Solutions) For each \((a, m)\),
\begin{enumerate}
    \item There exists a function $h_1$ such that
    \[
    p(y \mid z,a,\vec{m})
    =
    \sum_{w} h_1(y,a,\vec{m},w)\, p(w \mid z,a,\vec{m})
    \]
    \item There exists a function $h_0$ such that
    \begin{align*}
    &\sum_{w,m} h_1(y,\vec{m},a',w)\, p(w,\vec{m} \mid z,a)
    \\
    &\qquad =
    \sum_{w} h_0(y,w,a',a)\, p(w \mid z,a)
    \end{align*}

\end{enumerate}
\end{asn}
Assumption~\ref{assump:bridge_var1} posits solutions to Fredholm integral equations of the first kind, which are well studied in functional analysis literature \citep{kress99linear}. In finite-support settings, solutions can be constructed easily under relevant matrix invertibility conditions; see Appendix~\ref{app:bridge_set1}. These invertibility conditions state that distributions in $\{p(W \mid z,a,\vec{m}) : z \in \mathcal{Z}\}$ exhibit sufficient variation across values $z$ in its support $\mathcal{Z}$ for each $(a,\vec{m})$, and that distributions in \(\{p(W \mid z,a) : z \in \mathcal{Z}\}\) likewise vary sufficiently across values $z$ for each $a$.

Given these assumptions, we have the following result.
\begin{thm}\label{thm:id_set1}
Under Assumptions~\ref{assump:ci_var1}--\ref{assump:bridge_var1}, \(p(Y(a))\) is identified in Fig.~\ref{fig:primal_v1} by Equation~\ref{eq:id_set1} for every value $y$ as follows

\begin{equation}\label{eq:id_set1}
p(Y(a)=y)
\;=\;
\sum_{a',w} h_0(y, w,a',a)\, p(w,a')
\end{equation}
\end{thm}
See Appendix~\ref{app:id_set1} for a proof.

This result re-expresses {Equation} \ref{eq:front-door2} such that the first two terms correspond to a bridge function $h_0$, and the last term replaces marginalization over the (unobserved) $U$ with a marginalization over the (observed) proxy $W$.  

\subsection{Assumption Set 2}\label{subsec:id_set2}

The second identification strategy relies on the conditional independence relations in Assumption~\ref{assump:ci_var2}, reflected in the structure of Fig.~\ref{fig:primal_v2}, together with the structural conditions in Assumptions~\ref{assump:completeness_var2}--\ref{assump:bridge_var2}. 

\begin{asn}\label{assump:ci_var2} 
(Conditional Independences)
\begin{itemize}[noitemsep]
\item[1.] $Y \ci Z \mid \vec{U}, \vec{M}, A$
\item[2.] $W \ci \vec{M}, Z \mid \vec{U}, A$
\item[3.] $Z \ci W, A \mid \vec{U}$
\end{itemize}
\end{asn}

Assumption~\ref{assump:ci_var2} 
may be represented graphically via the m-separation criterion in the variant of the composite bow graph shown in Fig.~\ref{fig:primal_v2}.  Note that 
Assumptions \ref{assump:ci_var1} and \ref{assump:ci_var2} represent non-nested models, meaning neither is more general than the other.


\begin{asn}\label{assump:completeness_var2}
(Completeness) For each \((a, \vec{m})\) and any square integrable \(g\),
    \begin{itemize}[noitemsep]
        \item[1.] $\mathbb{E}\{g(\vec{u}) \mid z,a,\vec{m}\}=0$ a.s. iff $g(\vec{u})=0$ a.s. 
        \item[2.] $\mathbb{E}\{g(\vec{u}) \mid w,a\}=0$ a.s. iff $g(\vec{u})=0$ a.s. 
    \end{itemize}
\end{asn}
Just as Assumption~\ref{assump:completeness_var1},
Assumption~\ref{assump:completeness_var2} imposes distributional variation restrictions, so that distributions in \(\{p(Z \mid \vec{u}_i, a,\vec{m}): \vec{u}_i\}\) exhibit sufficient variation across values \(\vec{u}_i\) for each \((a,\vec{m})\), and distributions in \(\{p(W \mid u_i, a): \vec{u}_i\}\) exhibit sufficient variation across values \(\vec{u_i}\) for each \(a\).

We also require the following analogue to Assumption~\ref{assump:bridge_var1}.
\begin{asn}\label{assump:bridge_var2}
(Fredholm Equation Solutions) For each \((a, m)\),
\begin{enumerate}
    \item There is a function $b_1$ such that 
    \[p(y \mid z,a,\vec{m}) = \sum_{w} b_1(y,a,\vec{m},w)p(w \mid z,a,\vec{m})\]
    \item There is a function $b_0$ such that 
    \[p(\vec{m} \mid w,a) = \sum_{z} b_0(\vec{m},a,z)p(z \mid w,a)\]
\end{enumerate}
\end{asn}

Given these assumptions, we have the following result.

\begin{thm}\label{thm:id_set2}
Under Assumptions~\ref{assump:ci_var2}--\ref{assump:bridge_var2}, \(p(Y(a))\) is identified in Fig.~\ref{fig:primal_v2}
for every value $y$ as follows
\begin{align*}
p(Y(a)\!=\!y)
\!=\!\!
\sum_{\vec{m},a',w,z}
\!\!
b_1(y,a',\vec{m}, w)b_0(\vec{m}, a, z)p(w,z,a')
\end{align*}

\end{thm}
See Appendix~\ref{app:id_set2} for a proof.

This results replaces the first two terms in {Equation} \ref{eq:front-door2} bridge functions $b_0$ and $b_1$, with the resulting functional bearing close structural similarity to Equation \ref{eq:front-door2}. Assumptions~\ref{assump:completeness_var2} and~\ref{assump:bridge_var2} aid in this replacement by ensuring that
\begin{equation*}
p(y  \mid  \vec{u},a,\vec{m}, \cancel{z}) = \sum_{w} b_1(y,a,\vec{m},w)p(w  \mid  \vec{u}, a, \cancel{\vec{m}, z}),
\end{equation*}
and 
\begin{equation*}
p(\vec{m}  \mid  \vec{u},a, \cancel{w}) = \sum_{z} b_0(\vec{m}, a, z)p(z  \mid  \vec{u}, \cancel{w,a}),
\end{equation*}
where the independences on the left and right hand sides follow by Assumption~\ref{assump:ci_var2}.
With these replacements, summing over \(\vec{U}\) in Equation~\ref{eq:front-door2} eliminates its appearance from the expression.

In finite-support settings, solutions to the integral equations in Assumption~\ref{assump:bridge_var2} can be constructed easily under invertibility of relevant probability matrices; see Appendix~\ref{app:bridge_set2}. These invertibility conditions state that distributions in $\{p(W \mid z,a,\vec{m}): z \in \mathcal{Z}\}$ exhibit sufficient variation across values $z$ in its support $\mathcal{Z}$ for each $(a,\vec{m})$, and that distributions in $\{p(Z \mid w,a): w \in \mathcal{W}\}$ likewise vary sufficiently across values $w$ for each $a$. If the cardinalities of \(Z\) and \(W\) are misaligned so that these invertibility conditions fail, a coarsened version of the variables may be employed instead.

\subsection{Assumption Set 3}\label{subsec:id_set3}

The third identification strategy assumes conditional independence relations in Assumption~\ref{assump:ci_kp}, displayed via m-separation in Fig.~\ref{fig:primal_v3}, and Assumptions~\ref{assump:bounded}--\ref{assump:distinctness_var3}.




\begin{asn}\label{assump:ci_kp} 

(Conditional Independences)
\begin{itemize}[noitemsep]
\item[1.] \(W\), \(Z\), and \(Y\) are mutually independent \(\mid \vec{M},A,\vec{U}\)
\item[2.] $W \ci \vec{M}, A \mid \vec{U}$
\end{itemize}
\end{asn}
These independence assumptions can be displayed graphically via m-separation in the variant of the composite bow graph shown in Fig.~\ref{fig:primal_v3}.

\begin{asn}\label{assump:bounded}
(Bounded Density)
Random variables $(A,\vec{M},Y,W,Z,\vec{U})$ admit a joint density that is bounded. 
\end{asn}


\begin{asn}\label{assump:completeness_var3}
(Completeness) For each \((a,\vec m)\) and any square integrable \(g\), 
\begin{itemize}[noitemsep]
\item[1] $\mathbb{E}\{g(\vec{u}) \mid z,a,\vec m\}=0$ a.s. iff $g(\vec{u})=0$ a.s. 
\item[2] $\mathbb{E}\{g(\vec{u}) \mid w,a,\vec m\}=0$ a.s. iff $g(\vec{u})=0$ a.s.
\end{itemize}
\end{asn}

\begin{asn}\label{assump:distinctness_var3}
(Distinctness)
For each \((a,\vec m)\) and \(i \neq j\), \(p(Y \mid \vec{u}_i, a, \vec{m})\) and \(p(Y \mid \vec{u}_j, a, \vec{m})\) differ with positive probability under the marginal distribution of \(Y\).
\end{asn}

Assumption~\ref{assump:completeness_var3} imposes distributional variation restrictions, so that distributions in \(\{p(Z \mid \vec{u}_i): \vec{u}_i,a,\vec{m}\}\) exhibit sufficient variation across values \(\vec{u}_i\) for each \((a,\vec{m})\), and that distributions in \(\{p(W \mid u_i,a,\vec{m}): \vec{u}_i\}\) exhibit sufficient variation across values \(\vec{u_i}\) for each \((a,\vec{m})\). Assumption~\ref{assump:distinctness_var3} states that distributions in \(\{p(Y \mid \vec{u}_i, a,\vec{m}): \vec{u}_i\}\) differ across values \(\vec{u}_i\).

\begin{thm}\label{thm:id_set3}
Under Assumptions~\ref{assump:ci_kp}-\ref{assump:distinctness_var3}, \[\{ p(A,\vec{M},Y,\vec{u}_i,W,Z): \vec{u}_i \} \] is identified up to labeling \(i\). Hence, the counterfactual distribution in Figure~\ref{fig:primal_v3} is identified via Equation~\ref{eq:front-door2}.
\end{thm}
See Appendix~\ref{app:id_set3} for a proof, and Appendix~\ref{app:id_set3_discrete} for an example in a finite-support model that illustrates the main idea. When $U$ has finite support, Assumptions~\ref{assump:ci_kp}--\ref{assump:distinctness_var3} uniquely determine the size of the support of \(U\) \citep{Derksen2013}. A corresponding result holds in real-valued continuous models \citep{hu2008instrumental}. Note that while the model depicted in Fig.~\ref{fig:primal_v3} is nested within those depicted in Fig.~\ref{fig:primal_v1} and~\ref{fig:primal_v2}, the above theorem allows identification of the entire full law, and thus targets of inference other than the ACE.

This identification strategy extends eigendecomposition-based methods for causal 
identification under hidden confounding from \citet{kuroki2014measurement} and \citet{allman2015parameter} by incorporating identification arguments developed by \citet{hu2008instrumental} for continuous state spaces.

%% file: K_Estimation.tex
\section{Estimation}\label{sec:estimation}

Having described identifying functionals in Theorems \ref{thm:id_set1},\ref{thm:id_set2}, and \ref{thm:id_set3}, we now provide plug-in estimation strategies for all three functionals, and an influence function-based estimator for the third functional.

\subsection{Consistent Estimation under Assumption Sets 1, 2, and 3}\label{subsec:est_set123}

We construct 
$\hat \psi^{(S1, PI)},\hat \psi^{(S2, PI)},\hat \psi^{(S3, PI)}$,
the plug-in estimators for the counterfactual mean $\psi_a = \mathbb{E}[Y(a)]$ under Assumption Sets 1, 2, and 3, respectively,
which are consistent under standard regularity conditions. 

Under Assumption Sets 1 and 2, estimation proceeds by first estimating a mean-induced bridge (corresponding to $Y$-marginalized solution(s) of integral equation(s) in Equation~\ref{eq:mean_bridge1} and~\ref{eq:mean_bridge2} in Appendix~\ref{app:plug_in_estimation}), and then taking a sample average over the empirical distribution. Under Assumption Set 3, estimation proceeds by plugging in estimators of each distributional component in Equation~\ref{eq:front-door2}, obtained by solving a linear system of equations by matrix inversion. Further details for these estimators can be found in Appendix~\ref{app:plug_in_estimation}.

\subsection{Influence Function-based Estimation under Assumption Set 3}\label{subsec:est_set3}

Recall by Theorem~\ref{thm:id_set3}, under Assumptions~\ref{assump:ci_kp}-\ref{assump:distinctness_var3}, \[\{ p(A,\vec{M},Y,\vec{u}_i,W,Z): \vec{u}_i \} \] is identified up to labeling \(i\), giving the counterfactual distribution in Fig.~\ref{fig:primal_v3} via Equation~\ref{eq:front-door2}.

We propose an influence function-based estimator $\hat\psi^{(S3,IF)}_a$ of the counterfactual mean $\psi_a = \mathbb{E}[Y(a)]$. We now briefly review estimation theory in semiparametric models that we use to construct this estimator.

Let $P$ denote the data-generating distribution, $\psi(P)$ be the parameter of interest, $\eta$ be the collection of nuisance functions indexing $P$, and $O$ denote observed variables. An \emph{influence function} $\varphi(O;\psi(P),\eta_0 = \eta)$ for $\psi(P)$ at $P$ satisfies
\begin{equation}\label{eq:unbiased_correct}
\mathbb{E}_{P}\!\left[\varphi\!\left(O;\psi(P),\eta_0\right)\right] = 0 ,
\end{equation}
and for every regular parametric submodel $\{P_\epsilon:\epsilon\in\mathbb{R}\}$ through $P$ with $P_0=P$ and score
\begin{align*}
s_\epsilon(O)
:=
\left.\frac{d}{d\epsilon}\log p_\epsilon(O)\right|_{\epsilon=0},
\end{align*}
it holds that
\begin{equation}\label{eq:if_def_correct}
\left.\frac{d}{d\epsilon}\psi(P_\epsilon)\right|_{\epsilon=0}
=
\mathbb{E}_{P}\!\left[\varphi\!\left(O;\psi(P),\eta_0\right)\, s_\epsilon(O)\right]
\end{equation}
\citep{tsiatis06missing, KennedyTutorial}.
An influence function-based estimator \(\hat \psi\) solves estimating equation 
\begin{equation}\label{eq:est_equation}
\mathbb{E}_n[\varphi(O;\hat\psi, \hat \eta)] = 0
\end{equation}
These estimators can attain $\sqrt{n}$-consistency for $\psi_a$ despite nuisance functions \(\eta\) being estimated at slower than $\sqrt{n}$ rates \citep{chernozhukov_2018,KennedyTutorial}.

In addition to this desirable property, the estimator we derived, which we denote \(\hat \psi^{(S3,IF)}_a\), can exhibit \emph{multiple robustness}, meaning it can remain consistent for $\psi_a$ under misspecification of certain collections {of} nuisance functions (Theorem~\ref{thm:multiple_robustness}).  These properties extend to causal effect estimation (via \(\hat \psi^{(S3,IF)} = \hat \psi^{(S3,IF)}_1 - \hat \psi^{(S3,IF)}_0\)).

Our strategy for obtaining this estimator is to start with the nonparametric full data influence function (e.g. assuming the confounder $U$ is observed) for the counterfactual mean $\psi_a = \mathbb{E}[Y(a)]$, and then transform this object to account for the fact $U$ is not, in fact, observed.

\subsection{
The Full Data Influence Function
}

In the case where \(A,\vec{M},\vec{Y}, \vec{U}\) 
in Fig.~\ref{fig:primal_v0} 
are observed,
\cite{fulcher2020robust} derived the influence function for $\psi_a$, \[\varphi_{full}((A,\vec{M},\vec{Y}, \vec{U}); \cdot)\] shown in Equation~\ref{eq:full_eif}.
\begin{align}\label{eq:full_eif}
\varphi_{full}(.)
&=
\underbrace{
\frac{p(\vec{M} \mid A=a, \vec{U})}{p(\vec{M} \mid A, \vec{U})}
\Big\{
  Y - \mu(\vec{M}, A, \vec{U})
\Big\}
}_{\varphi_{full_1}(Y,\vec{M},A,\vec{U})}
\\
\notag
&+
\underbrace{
\frac{\mathbb{I}(A = a)}{\pi(A \mid \vec{U})}
\Big\{
  \xi(\vec{M}, \vec{U}) - \theta(\vec{U})
\Big\}
}_{\varphi_{full_2}(\vec{M},A,\vec{U})}
+\hspace{-0.2cm}
\underbrace{
  \eta(A,\vec{U}) 
}_{\varphi_{full_3}(A,\vec{U})} \hspace{-0.1cm}- \psi_{a},
\end{align}
where
\[
\mu(\vec{M},A,\vec{U})
:= \mathbb{E}[Y \mid \vec{M}, A, \vec{U}],
\hspace{0.2cm}
\pi(A \mid \vec{U})
:= p(A \mid \vec{U}),
\]
\[
\xi(\vec{M}, \vec{U})
:= \sum_{a=0}^1 \mu(\vec{M}, a, \vec{U})\,\pi(a \mid \vec{U}),
\]
\[
\eta(A, \vec{U})
:= \sum_{\vec{m}} \mu(\vec{m}, A, \vec{U})\, p(\vec{m} \mid a, \vec{U})\,
\]
\[
\theta(\vec{U})
:= \sum_{\vec{m}} \xi(\vec{m}, \vec{U})\, p(\vec{m} \mid a, \vec{U})\,
\]
After introducing \(W\) and \(Z\) (as in Fig.~\ref{fig:primal_v3}), 
\(\varphi_{\mathrm{full}}(\cdot)\) from Equation~\ref{eq:full_eif} remains an influence function for \(\psi_a\) in the fully observed model (where \(A, \vec{M}, \vec{Y}, \vec{U}, W, Z\) are observed). See Appendix~\ref{app:full_if} for details.

\subsection{
The Observed Data Influence Function
}

Under Assumptions~\ref{assump:ci_kp}-\ref{assump:distinctness_var3} in the model with hidden $\vec{U}$ in Fig.~\ref{fig:primal_v3}, we construct an influence function for \(\psi_a\),
\[
\varphi_{\mathrm{obs}}((A,\vec{M},Y,W,Z);\psi,\eta)
\] 
shown in Equation~\ref{eq:obs_if}.

Our construction exploits the fact that the full data distribution, and thus any functional of it, such as
\(
\varphi_{\mathrm{full}}(A,\vec{M},Y,\vec{U}=\vec{u}_i; \cdot),
\)
is identified up to a relabeling of the latent states \(\vec{u}_i\), and can be estimated. We express \(\varphi_{\mathrm{obs}}((A,\vec{M},Y,W,Z);\cdot)\) as a weighted combination of the terms in \(\varphi_{\mathrm{full}}(A,\vec{M},Y,\vec{U}=\vec{u}_i; \cdot)\) across latent states. The weights are constructed from observed variables and their moments conditional on values of $\vec{U}$. 

\subsubsection{A Simple Example: Binary $U$, Real-Valued $W,Z$}

To illustrate the construction strategy for the weights in our influence function,
consider the simple case of binary $U$ and real-valued $W,Z$. Define for $i \neq j$,
\begin{align*}
&\mathcal{C}_{WZ}(u_i,a)
=
\frac{W-\mathbb{E}[W \mid u_j]}{\mathbb{E}[W \mid u_i]-\mathbb{E}[W \mid u_j]} \times\\
&\frac{Z-\mathbb{E}[Z \mid u_j,a]}{\mathbb{E}[Z \mid u_i,a]-\mathbb{E}[Z \mid u_j,a]}. 
\end{align*}
Define \(\mathcal{C}_{YW}(u_i,a,\vec{m})\) and \(\mathcal{C}_{YZ}(u_i,a,\vec{m})\) similarly. 
For instance, 
\begin{align*}
&\mathcal{C}_{YW}(u_i,a,\vec{m})
=
\frac{Y-\mathbb{E}[Y \mid u_j,a,\vec{m}]}{\mathbb{E}[Y \mid u_i,a,\vec{m}]-\mathbb{E}[Y \mid u_j,a,\vec{m}]}
\,\times \\
& \qquad
\frac{W-\mathbb{E}[W \mid u_j]}{\mathbb{E}[W \mid u_i]-\mathbb{E}[W \mid u_j]},  i \neq j.
\end{align*}
Lastly, define \(\mathcal{C}_{WZY}(u_i,a,\vec{m})\) as follows:
\begin{align*}
&\mathcal{C}_{YWZ}(u_i,a,\vec{m})
=
\frac{Y-\mathbb{E}[Y \mid u_j,a,\vec{m}]}{\mathbb{E}[Y \mid u_i,a,\vec{m}]-\mathbb{E}[Y \mid u_j,a,\vec{m}]}
\,\times \\
& \qquad
\frac{W-\mathbb{E}[W \mid u_j]}{\mathbb{E}[W \mid u_i]-\mathbb{E}[W \mid u_j]}
\frac{Z-\mathbb{E}[Z \mid u_j,a]}{\mathbb{E}[Z \mid u_i,a]-\mathbb{E}[Z \mid u_j,a]},  i \neq j.
\end{align*}

These weights require that conditional means in weight denominators be distinct, not already implied by Assumption~\ref{assump:completeness_var3}. Minor discrepancies between identification and estimation assumptions are common; e.g., augmented inverse probability weighting assumes a discrete treatment for effect estimation, when it is not required for identification.



Given these weights, which may be viewed as \emph{Lagrange interpolating polynomials,}
Theorem~\ref{thm:obs_if} in the next section gives the desired influence function.

\subsubsection{The General Case: Arbitrary Hidden $\vec{U}$ with Finite Support}

The simple Lagrange interpolating polynomial weights described in the previous section may be generalized, for a latent variable \(\vec{U}\) with
finitely many levels, to weights  
\(
\mathcal{C}_{WZ}(\vec{u}_i,a)\), \(\mathcal{C}_{YW}(\vec{u}_i,\!a,\!\vec{m})\), \(\mathcal{C}_{YZ}(\vec{u}_i,\!a,\!\vec{m})\), 
and \(\mathcal{C}_{YWZ}(\vec{u}_i,\!a,\!\vec{m})\)
described in Appendix~\ref{app:proxy_contrasts}. These weights are constructed to satisfy Equation~\ref{eqn:weight_prop1} and~\ref{eqn:weight_prop2}.
\begin{equation}\label{eqn:weight_prop1}
\begin{aligned}
&\text{ For each } X \in \{W,Z\}, \\
&\mathbb{E}\big[
\mathcal{C}_{WZ}(A,\vec{u}_i) \mid Y, \vec{M}, A, \vec{u}_i,X\big] =  k,  \quad k \text{ is a constant}; \\
&\mathbb{E}\big[
\mathcal{C}_{WZ}(A,\vec{u}_i) \mid Y, \vec{M}, A, \vec{u}_j,X\big] =  0, \quad  \vec{u}_i \neq \vec{u}_j.
\end{aligned}
\end{equation}

\begin{equation}\label{eqn:weight_prop2}
\begin{aligned}
&\text{For each } X \in \{W,Z,Y\},\\
&\mathbb{E}\big[
\mathcal{C}_{WZ}(A,\vec{u}_i) + \mathcal{C}_{YW}(\vec{M},A,\vec{u}_i) + \mathcal{C}_{YZ}(\vec{M},A,\vec{u}_i) \\
&-2 \mathcal{C}_{WZY}(\vec{M},A,\vec{u}_i)\mid \vec{M}, A, \vec{u_i},X\big] =  1; \\
&\mathbb{E}\big[
\mathcal{C}_{WZ}(A,\vec{u}_i) + \mathcal{C}_{YW}(\vec{M},A,\vec{u}_i) + \mathcal{C}_{YZ}(\vec{M},A,\vec{u}_i) \\
&-2 \mathcal{C}_{WZY}(\vec{M},A,\vec{u}_i)\mid \vec{M}, A,\vec{u_j},X\big] =  0, \quad \vec{u}_i \neq \vec{u}_j.
\end{aligned}
\end{equation}

Developing the analogous construction of these weights beyond the setting where \(U\) has finite support remains an open problem, which we leave for future work. Nevertheless, if Equations~\ref{eqn:weight_prop1} and~\ref{eqn:weight_prop2} hold, we have the following result.

\begin{thm}\label{thm:obs_if} In the model 
in Fig.~\ref{fig:primal_v3} (where $\vec{U}$ is unobserved), under Assumptions~\ref{assump:ci_kp}--\ref{assump:distinctness_var3}, a nonparametric observed data influence function for $\psi_a$ is given by
\begin{align}\label{eq:obs_if}
\notag
&\varphi_{\mathrm{obs}}\!\big((A,\vec{M},Y,W,Z);\psi_a,\eta\big) \\
\notag
&= \sum_{i}
\mathcal{C}_{WZ}(A,\vec{u}_i)\,
\varphi_{\mathrm{full}_1}(Y,\!\vec{M}\!,\!A,\!\vec{u}_i) + 
\sum_i
\Big\{
\mathcal{C}_{WZ}(A,\vec{u}_i) \\
\notag
&+\mathcal{C}_{YW}(\vec{M}\!,\!A,\!\vec{u}_i)  +
\mathcal{C}_{YZ}(\vec{M}\!,\!A,\!\vec{u}_i) -2\,\mathcal{C}_{WZY}(\vec{M}\!,\!A,\!\vec{u}_i)
\Big\} \\
&\cdot \Big\{
\varphi_{\mathrm{full}_2}(\vec{M}\!,\!A,\!\vec{u}_i) + \varphi_{\mathrm{full}_3}(A,\!\vec{u}_i)
\Big\} - \psi_{a},
\end{align}
 where the terms \(\varphi_{\mathrm{full}_1}(Y,\vec{M},A,\vec{u}_i)\), \(\varphi_{\mathrm{full}_2}(\vec{M},A,\vec{u}_i)\), and \(\varphi_{\mathrm{full}_3}(A,\vec{u}_i)\) are defined in Equation~\ref{eq:full_eif}, and the weights \(
\mathcal{C}_{WZ}(A,\vec{u}_i)\), \(\mathcal{C}_{YW}(\vec{M},A,\vec{u}_i)\), \(\mathcal{C}_{YZ}(\vec{M},A,\vec{u}_i)\), and 
\(\mathcal{C}_{YWZ}(\vec{M},A,\vec{u}_i)\) are constructed to satisfy Equation~\ref{eqn:weight_prop1} and~\ref{eqn:weight_prop2}.
\end{thm}
The proof appears in Appendix~\ref{app:obs_if}.


As is the case for many influence function based estimators, the estimator \(\hat \psi^{(S3,IF)}_a\) obtained by solving Equation~\ref{eq:est_equation} using the influence function in Equation~\ref{eq:obs_if} can exhibit multiple robustness.  Specifically, we have the following result (proved in Appendix~\ref{app:asymp_normality}).

\begin{thm}\label{thm:multiple_robustness}
Let \(\varphi(O;\psi_a,\eta)\) be defined as in Theorem~\ref{thm:obs_if}. 
The estimator \(\hat{\psi}^{(S3,IF)}_a\), obtained as the solution to 
Equation~\ref{eq:est_equation}, is consistent for \(\psi_a\) if any one of the 
nuisance sets (1)--(5) is correctly specified.
\begin{itemize}[noitemsep]
  \item[1.] \(p(\vec{M} \mid A, \vec{U})\), \(\{\mathbb{E}[h_j(W) \mid \vec{U}]\}_j\), \(\{\mathbb{E}[f_j(Z) \mid \vec{U},A]\}_j\);  
  \item[2.] \(p(\vec{M} \mid A, \vec{U})\), \(\{\mathbb{E}[h_j(W) \mid \vec{U}]\}_j\), \\ \(\{\mathbb{E}[t_j(Y) \mid \vec{M}, A, \vec{U}]\}_j\);  
  \item[3.] \(f(\vec{M} \mid A, \vec{U})\), \(\{\mathbb{E}[f_j(Z) \mid \vec{U},A]\}_j\), \\ \(\{\mathbb{E}[t_j(Y) \mid \vec{M},A,\vec{U}]\}_j\);
  \item[4.] \(\{\mathbb{E}[t_j(Y) \mid \vec{M},A,\vec{U}]\}_j\), \(\pi(A \mid \vec{U})\), \(\{\mathbb{E}[h_j(W) \mid \vec{U}]\}_j\);  
  \item[5.] \(\{\mathbb{E}[t_j(Y) \mid \vec{M},A,\vec{U}]\}_j\), \(\pi(A \mid \vec{U})\), \(\{\mathbb{E}[f_j(Z) \mid \vec{U},A]\}_j\),
\end{itemize}
where $\{h_j(W)\}_j$, $\{f_j(Z)\}_j$, and $\{t_j(Y)\}_j$ are collections defined in Appendix~\ref{app:proxy_contrasts}. By definition, these collections include the identity functions on \(W\), \(Z\), and \(Y\). When \(\vec U\) is binary, the collections reduce to the identity functions.
\end{thm}

Moreover, \(\hat \psi^{(S3,IF)}_a\) is asymptotically normal and \(\sqrt{n}\)-consistent despite slower than parametric rates of convergence for nuisance functions. We have the following result (proved in Appendix~\ref{app:asymp_normality}).

\begin{thm}\label{thm:asymp_normality}
Let \(\varphi(O;\psi_a,\eta)\) be defined as in Theorem~\ref{thm:obs_if}. 
Then if (1)-(5) below hold (for all \(j\) where applicable), under regularity conditions including sample splitting and consistency of all nuisance functions, the estimator \(\hat{\psi}^{(S3,IF)}_a\) obtained as the solution to Equation~\ref{eq:est_equation} has the property
\[
\sqrt{n}\,(\hat{\psi}_a - \psi_a)
\;\;\overset{d}{\longrightarrow}\;\;
\mathcal{N}\!\Big(0, \,\mathbb{E}\big[\varphi^2(O; \psi_a, \eta)\big]\Big).
\]
\begin{itemize}[noitemsep]
  \item[1.] $\lVert \hat{\mu}(\vec{M}, A, \vec{U}) - \mu(\vec{M}, A, \vec{U}) \rVert$ \\ 
  $\lVert \hat{p}(\vec{M} \mid A, \vec{U}) - p(\vec{M} \mid A, \vec{U}) \rVert = o_p(n^{-1/2})$,
  \item[2.] $\lVert \hat{\pi}(A \mid \vec{U}) - \pi(A \mid \vec{U}) \rVert$ \\
  $\lVert \hat{p}(\vec{M} \mid A, \vec{U}) - p(\vec{M} \mid A, \vec{U}) \rVert= o_p(n^{-1/2})$;
\item[3.] $\lVert \hat{\mathbb{E}}[t(Y)_j \mid \vec{M}, A, \vec{U}] - \mathbb{E}[t(Y)_j \mid \vec{M}, A, \vec{U}] \rVert$ \\
  $\lVert \hat{\mathbb{E}}[h(W)_j \mid \vec{U}] - \mathbb{E}[h(W)_j \mid \vec{U}] \rVert = o_p(n^{-1/2})$,
 \item[4.] $\lVert \hat{\mathbb{E}}[t(Y)_j \mid \vec{M}, A, \vec{U}] - \mathbb{E}[t(Y)_j \mid \vec{M}, A, \vec{U}] \rVert$ \\
  $\lVert \hat{\mathbb{E}}[f(Z)_j \mid \vec{U},A] - \mathbb{E}[f(Z)_j \mid \vec{U} ,A]\rVert = o_p(n^{-1/2})$,
 \item[5.] $\lVert \hat{\mathbb{E}}[f(Z)_j \mid \vec{U},A] - \mathbb{E}[f(Z)_j \mid \vec{U},A]\rVert$ \\
  $\lVert \hat{\mathbb{E}}[h(W)_j \mid \vec{U}] - \mathbb{E}[h(W)_j \mid \vec{U} ,A]\rVert = o_p(n^{-1/2})$,
\end{itemize}
where $\{h_j(W)\}_j$, $\{f_j(Z)\}_j$, and $\{t_j(Y)\}_j$ are collections  defined in Appendix~\ref{app:proxy_contrasts}. By definition, these collections include the identity functions on \(W\), \(Z\), and \(Y\). When \(\vec U\) is binary, the collections reduce to the identity functions.
\end{thm}

%% file: L_Experiments.tex
\section{Experiments }\label{sec:experiments}
We demonstrate 
consistency of our proposed estimators using
a simulation study for a finite-support data-generating process (DGP) with binary variables in Table~\ref{table: results}, and also for a mixed DGP in Table~\ref{table: results_cont}. To ensure the datasets generated satisfy all three sets of assumptions in Section~\ref{sec:identification}, we selected models in their intersection. See Appendix \ref{app:DGP} and Appendix~\ref{app:DGP2} for details of the finite-support DGP and mixed DGP, respectively.

We consider the following estimators (definitions can be found in Appendix \ref{app:plug_in_estimation} and \ref{app:bm_estimators}): 
\begin{itemize}[noitemsep]
    \item[1.] The oracle plug-in estimator \(\hat{\psi}^{(\mathrm{Oracle,PI})}\)that has access to the full data, where variables \(A, \vec{M},Y,\vec{U}, W, Z\) are all observed (Equation~\ref{eq:S1_Oracle});
    \item[2.] A front-door plug-in estimator \(\hat{\psi}^{(\mathrm{FD,PI})}\) that (incorrectly) assumes \(\vec{M}\) satisfies the front-door criterion relative to \((A,Y)\) (Equation~\ref{eq:S3_FD});  
    \item[3.] The plug-in estimators \(\hat{\psi}^{(S1,PI)}\),\(\hat{\psi}^{(S2,PI)}\), \(\hat{\psi}^{(S3,PI)}\) (Equations \ref{eq:S1_PI},\ref{eq:S2_PI}, and \ref{eq:S3_PI}), and the influence function-based estimator \(\hat{\psi}^{(S3,IF)}\) (obtained by solving Equation~\ref{eq:est_equation} using influence function in Equation~\ref{eq:obs_if});
    \item[4.] The estimators $\hat{\psi}^{(S3,PI,MIS)}$, $\hat{\psi}^{(S3,IF, MIS)}$ which mirror the formulation of \(\hat{\psi}^{(S3,PI)}\) and \(\hat{\psi}^{(S3,IF)}\) but misspecify $p(M|a,u)$ for each $a$ value by labeling
    $u_i$ incorrectly.  
\end{itemize}

Results for a finite-support DGP are shown in Table.\ref{table: results}. At relatively small sample size ($\leq$ 1000), all non-oracle estimators perform poorly. 
Additionally, $\hat{\psi}^{(S3,IF)}$ exhibits
instability. Estimator instability is a known issue with many inverse weighted estimators arising in semi-parametric theory, even in cases where proximal inference is not involved. In Table~\ref{table: results}, the result for \(\hat \psi^{(S3, IF, Clipped)}\) illustrates how clipping can improve estimator stability at low sample size. At a moderately large sample size ($\geq 3000$), the front-door estimator $\hat{\psi}^{(FD,PI)}$ still exhibits substantial bias, due to model misspecification, while
\(\hat{\psi}^{(S1,PI)}\),\(\hat{\psi}^{(S2,PI)}\),\(\hat{\psi}^{(S3,PI)}\), and \(\hat{\psi}^{(S3,IF)}\) all provide consistent estimates. 

To demonstrate the robustness property of the influence function-based estimator $\hat{\psi}^{(S3,IF)}$,
we misspecified $p(M\mid a,u)$ for each $a$ value by labeling
$u_i$ incorrectly and recording
the resulting estimates
as $\hat{\psi}^{(S3,PI,MIS)}$ and $\hat{\psi}^{(S3,IF, MIS)}$. Unlike plug-in estimator $\hat{\psi}^{(S3,PI,MIS)}$, which outputs a biased estimate, the influence function-based estimator $\hat{\psi}^{(S3,IF, MIS)}$ remains consistent.

\begin{table}[t]
\caption{Simulation results for the binary DGP$^*$; Ground truth ACE = 0.258.} 
\centering
\begin{tabular}[h]{rlrrr}
\toprule
n & method & mean & bias & variance\\
\midrule
1000 & $\hat{\psi}^{(Oracle)}$ & 0.254 & -2.5e-03 & 7.0e-04\\
1000 & $\hat{\psi}^{(FD,PI)}$ & 0.159 & -9.7e-02 & 2.3e-04\\
1000 & $\hat{\psi}^{(S1,PI)}$ & 0.231 & -2.6e-02 & 1.5e-02\\
1000 & $\hat{\psi}^{(S2,PI)}$ & 0.237 & -2.0e-02 & 9.5e-03\\
1000 & $\hat{\psi}^{(S3,PI)}$ & 0.163 & -9.3e-02 & 8.7e-03\\
1000 & $\hat{\psi}^{(S3,IF)}$ & 2.6e+11 & 2.6e+11 & 3.3e+23\\
1000 & $\hat{\psi}^{(S3,IF, Clipped)^{\dagger}}$ & 0.333 & 7.6e-02 & 1.6e-01\\ 
\specialrule{1.5pt}{0pt}{0pt}
\addlinespace

3000 & $\hat{\psi}^{(Oracle)}$ & 0.257 & 4.0e-04 & 2.0e-04\\
3000 & $\hat{\psi}^{(FD,PI)}$ & 0.161 & -9.6e-02 & 6.0e-05\\
3000 & $\hat{\psi}^{(S1,PI)}$ & 0.259 & 2.4e-03 & 7.2e-04\\
3000 & $\hat{\psi}^{(S2,PI)}$ & 0.259 & 1.9e-03 & 6.2e-04\\
3000 & $\hat{\psi}^{(S3,PI)}$ & 0.242 & -1.5e-02 & 3.7e-03\\
3000 & $\hat{\psi}^{(S3,IF)}$ & 0.261 & 3.9e-03 & 7.2e-03\\
\midrule
\addlinespace
3000 & $\hat{\psi}^{(S3,PI,MIS)}$ & -0.181 & -4.4e-01 & 4.0e-03\\
3000 & $\hat{\psi}^{(S3,IF, MIS)}$ & 0.243 & -1.4e-02 & 9.2e-03\\
\specialrule{1.5pt}{0pt}{0pt}
\addlinespace

6000 & $\hat{\psi}^{(Oracle)}$ & 0.257 & 6.0e-04 & 9.7e-05\\
6000 & $\hat{\psi}^{(FD,PI)}$ &
0.161 & -9.6e-02 & 3.1e-05\\
6000 & $\hat{\psi}^{(S1,PI)}$ & 0.257 & 1.5e-04 & 4.2e-04\\
6000 & $\hat{\psi}^{(S2,PI)}$ & 0.258 & 1.1e-03 & 4.2e-04\\
6000 & $\hat{\psi}^{(S3,PI)}$ & 0.249 & -8.1e-03 & 2.7e-03\\
6000 & $\hat{\psi}^{(S3,IF)}$ & 0.249 & -7.6e-03 & 2.5e-03\\
\midrule
\addlinespace
6000 & $\hat{\psi}^{(S3,PI,MIS)}$ & -0.198 & -4.5e-01 & 3.2e-03\\
6000 & $\hat{\psi}^{(S3,IF, MIS)}$ & 0.253 & -4.1e-03 & 8.1e-03\\
\specialrule{1.5pt}{0pt}{0pt}
\end{tabular}
\noindent
\begin{minipage}{0.95\textwidth}
\footnotesize
$*$ Data-generating process can be found in Appendix~\ref{app:DGP}. \\
$^\dagger$ Clipping was applied to
$\frac{Y-E[Y\mid u_j,a,\vec{m}]}
{E[Y\mid u_i,a,\vec{m}]-E[Y\mid u_j,a,\vec{m}]}$,
where the \\denominator's  magnitude was lower bounded by {0.01}.

\end{minipage}
\label{table: results}
\end{table}

\begin{table}[h]
\caption{Simulation results for the mixed DGP$^*$; Ground truth ACE = 0.169.} 
\centering
\begin{tabular}[h]{rlrrr}
\toprule
n & method & mean & bias & variance\\
\midrule
1000 & $\hat{\psi}^{(Oracle)}$ & 0.167 & -2.1e-03 & 1.0e-03\\
1000 & $\hat{\psi}^{(FD,PI)}$ & 0.303 & 1.3e-01 & 1.4e-03\\
1000 & $\hat{\psi}^{(S1,PI)}$ & 0.174 & 5.2e-03 & 3.1e-02\\
1000 & $\hat{\psi}^{(S2,PI)}$ & 0.175 & 6.0e-03 & 1.9e-02\\
1000 & $\hat{\psi}^{(S3,PI)}$ & 0.197 & 2.8e-02 & 1.9e-03\\
1000 & $\hat{\psi}^{(S3,IF)}$ & 0.161 &-7.7e-03 & 5.9e-03\\
\specialrule{1.5pt}{0pt}{0pt}
\addlinespace

3000 & $\hat{\psi}^{(Oracle)}$ & 0.170 & 7.0e-04 & 3.0e-04\\
3000 & $\hat{\psi}^{(FD,PI)}$ & 0.303 & 1.3e-01 & 7.0e-04\\
3000 & $\hat{\psi}^{(S1,PI)}$ & 0.173 & 3.8e-03 & 9.0e-04\\
3000 & $\hat{\psi}^{(S2,PI)}$ & 0.168 & -1.3e-03 & 1.3e-03\\
3000 & $\hat{\psi}^{(S3,PI)}$ & 0.180 & 1.1e-02 & 1.2e-03\\
3000 & $\hat{\psi}^{(S3,IF)}$ & 0.179 & 1.1e-02 & 2.8e-03\\
\midrule
\specialrule{1.5pt}{0pt}{0pt}
\addlinespace

6000 & $\hat{\psi}^{(Oracle)}$ & 0.169 & -2.0e-04 & 2.0e-04\\
6000 & $\hat{\psi}^{(FD,PI)}$ & 0.301 & 1.3e-01 & 3.0e-04\\
6000 & $\hat{\psi}^{(S1,PI)}$ & 0.164 & -4.6e-03 & 4.0e-04\\
6000 & $\hat{\psi}^{(S2,PI)}$ & 0.167 & -1.1e-03 & 5.0e-04\\
6000 & $\hat{\psi}^{(S3,PI)}$ & 0.168 & -9.0e-04 & 5.0e-04\\
6000 & $\hat{\psi}^{(S3,IF)}$ & 0.166 & -2.7e-03 & 1.5e-03\\
\midrule
\specialrule{1.5pt}{0pt}{0pt}
\end{tabular}
\begin{minipage}{0.95\textwidth}
\footnotesize
$*$ Data-generating process can be found in Appendix~\ref{app:DGP2}.
\end{minipage}
\label{table: results_cont}
\end{table}

Results for the mixed DGP are shown in Table~\ref{table: results_cont}. The estimators exhibit patterns similar to those observed in Table~\ref{table: results}, although small-sample bias is reduced. 

%% file: N_Discussion.tex
\section{Discussion}\label{sec:discussion}

The core contribution of our work is development of methods that 
address a central difficulty with the front-door criterion: that mediators often share unmeasured causes with the treatment and outcome.
Specifically, we 
develop three distinct identification strategies based on proxies 
which allow identification of causal effects
in cases where the front-door criterion fails. 
The first two strategies accommodate continuous hidden confounders and identify causal effects via solutions to Fredholm integral equations, relying on non-nested conditional independence assumptions that permit different causal structures. The third strategy permits a more restricted class of causal structures, but recovers the full law, allowing 
any functional of the full data distribution to be identified. 

In addition, we develop estimation strategies for the obtained identifying functionals based on the plug-in principle.  Finally, we develop an influence function based estimator for the functional obtained via the third strategy, based on estimation theory in semi-parametric models.

Simulation experiments demonstrate that all three approaches effectively debias causal effect estimation in the presence of hidden confounding when the front-door 
criterion fails. Moreover, the simulations illustrate that the influence function–based estimator for functionals obtained under the third strategy remains consistent even under certain forms of model misspecification.


%% file: S_Appendix.tex
\section{Identification Proof for Assumption Set 1}\label{app:id_set1}

\begin{proof}
Assumption~\ref{assump:ci_var1} and first bridge solution in Assumption~\ref{assump:bridge_var1} gives\

\begin{equation}
\sum_{\vec{u}}p(y \mid \vec{u},a',\vec{m},\cancel{z})p(\vec{u} \mid z, a',\vec{m})= \sum_{w, \vec{u}} h_1(y,a',\vec{m},w)p(w \mid u,\cancel{a',\vec{m},z})p(\vec{u} \mid z, a',\vec{m}).
\end{equation}

Then, the first completeness condition in Assumption~\ref{assump:completeness_var1} gives

\begin{equation}\label{eq:bridge1_id_set1}
p(y \mid \vec{u},a',\vec{m}) = \sum_{w} h_1(y,a',\vec{m},w)p(w \mid \vec{u}).
\end{equation}

Similarly, Assumption~\ref{assump:ci_var1} and the second bridge solution in Assumption~\ref{assump:bridge_var1} gives

\begin{equation}
\sum_{w,\vec{m}, \vec{u}}h_1(y,\vec{m},a',w)p(w,\vec{m} \mid \vec{u},\cancel{z},a)p(\vec{u} \mid z, a) = \sum_{w, \vec{u}}h_0(y, w,
a',a)p(w \mid \vec{u},\cancel{z,a})p(\vec{u} \mid z, a).
\end{equation}

Then, the second completeness condition in Assumption~\ref{assump:completeness_var1} gives

\begin{equation}\label{eq:bridge2_id_set1}
\sum_{w,\vec{m}}h_1(y,\vec{m},a',w)p(w,\vec{m} \mid \vec{u},a)= \sum_{w}h_0(y, w,
a',a)p(w \mid \vec{u}).
\end{equation}

Hence, we have that

\begin{align}
    p(Y(a)=y) &= \sum_{\vec{u},\vec{m},a'} p(y \mid \vec{m},a',\vec{u})p(\vec{m} \mid a,\vec{u})p(a' \mid \vec{u})p(\vec{u})\\
    &=\sum_{\vec{u},\vec{m},a', w}h_1(y,\vec{m},a',w)p(w \mid \vec{u})p(\vec{m} \mid a,\vec{u})p(a' \mid \vec{u})p(\vec{u}) \text{ by Equation~\ref{eq:bridge1_id_set1}}\\
    &=\sum_{\vec{u},\vec{m},a', w}h_1(y,\vec{m},a',w)p(w \mid \vec{u},a)p(\vec{m} \mid a,\vec{u})p(a' \mid \vec{u})p(\vec{u}) \text{ by Assumption~\ref{assump:ci_var1}}\\
    &= \sum_{\vec{u},a',\vec{m},w} h_1(y,\vec{m},a',w)p(w,\vec{m} \mid a,\vec{u})p(a' \mid \vec{u}) p(\vec{u}) \\
    &= \sum_{\vec{u},a',w} h_0(y, w,a',a)p(w \mid \vec{u})p(a' \mid \vec{u}) p(\vec{u}) \text{ by Equation~\ref{eq:bridge2_id_set1}}\\
    &= \sum_{\vec{u},a',w} h_0(y, w,a',a)p(w \mid \vec{u},a')p(a' \mid \vec{u}) p(\vec{u}) \text{ by Assumption~\ref{assump:ci_var1}}\\
    &= \sum_{\vec{u},a',w} h_0(y, w,a',a)p(w, a' \mid \vec{u})p(\vec{u}) \\
    &= \sum_{a',w}h_0(y, w,a',a)p(w,a')
\end{align}
\end{proof}

\section{Bridge Solutions for Assumption~\ref{assump:bridge_var1} in Finite-Support Models}\label{app:bridge_set1}

For each \((a, \vec{m})\), let \[\mathbf{P}_{Y \mid Z, a,\vec{m}} = \begin{bmatrix}
\Pr(y_1 \mid z_1, a,\vec{m}) & \cdots & \Pr(y_1 \mid z_{|\mathcal{Z}|}, a,\vec{m}) \\ 
\Pr(y_2 \mid z_1, a,\vec{m}) & \cdots & \Pr(y_2 \mid z_{|\mathcal{Z}|}, a,\vec{m}) \\ 
\vdots & \ddots & \vdots \\ 
\Pr(y_{|\mathcal{Y}|} \mid z_1, a,\vec{m}) & \cdots & \Pr(y_{|\mathcal{Y}|} \mid z_{|\mathcal{Z}|}, a,\vec{m}) 
\end{bmatrix}\] 

and \[\mathbf{P}_{W \mid Z, a,\vec{m}} = \begin{bmatrix}
\Pr(w_1 \mid z_1, a,\vec{m}) & \cdots & \Pr(w_1 \mid z_{|\mathcal{Z}|}, a,\vec{m}) \\ 
\Pr(w_2 \mid z_1, a,\vec{m}) & \cdots & \Pr(w_2 \mid z_{|\mathcal{Z}|}, a,\vec{m}) \\ 
\vdots & \ddots & \vdots \\ 
\Pr(w_{|\mathcal{W}|} \mid z_1, a,\vec{m}) & \cdots & \Pr(w_{|\mathcal{W}|} \mid z_{|\mathcal{Z}|}, a,\vec{m}) 
\end{bmatrix}.\]

Then if distributions in \(\{p(W \mid z,a,\vec{m}): z \in \mathcal{Z}\}\) vary sufficiently across values \(z\) in its support \(\mathcal{Z}\) for each \((a,\vec{m})\) so that \(\mathbf{P}_{W \mid Z, a,\vec{m}}\) is left invertible with left inverse \(\mathbf{P}_{W \mid Z, a,\vec{m}}^\dagger\), we have that

\begin{equation}
\mathbf{P}_{Y \mid Z, a,\vec{m}} = \mathbf{P}_{Y \mid Z, a,\vec{m}} \mathbf{P}_{W \mid Z, a,\vec{m}}^\dagger \mathbf{P}_{W \mid Z, a,\vec{m}}
\end{equation}

Hence, \( \mathbf{P}_{Y \mid Z, a,\vec{m}} \mathbf{P}_{W \mid Z, a,\vec{m}}^\dagger  \) defines \(h_1(y,a,\vec{m},w)\) in Assumption~\ref{assump:bridge_var1} such that

\[p(y \mid z,a,\vec{m}) = \sum_{w} h_1(y,a,\vec{m},w)p(w \mid z,a,\vec{m}).\]

Similarly, for each \(a\), let \[\mathbf{P}_{h_1(Y) \mid Z, a} = \begin{bmatrix}
 \sum_{w, \vec{m}} h_1(y_1,a,\vec{m},w)p(w, \vec{m} \mid z_1, a) & \cdots & \sum_{w, \vec{m}} h_1(y_1,a,\vec{m},w)p(w, \vec{m} \mid z_{||\mathcal{Z}|}, a) \\ 
\sum_{w, \vec{m}} h_1(y_2,a,\vec{m},w)p(w, \vec{m} \mid z_2, a)& \cdots & \sum_{w, \vec{m}} h_1(y_1,a,\vec{m},w)p(w, \vec{m} \mid z_{||\mathcal{Z}|}, a) \\ 
\vdots & \ddots & \vdots \\ 
\sum_{w, \vec{m}} h_1(y_{|\mathcal{Y}|} ,a,\vec{m},w)p(w, \vec{m} \mid z_1, a) & \cdots & \sum_{w, \vec{m}} h_1(y_{|\mathcal{Y}|} ,a,\vec{m},w)p(w, \vec{m} \mid z_{||\mathcal{Z}|}, a)
\end{bmatrix}\] 

Then if \(\mathbf{P}_{W \mid Z, a}\) is left invertible with left inverse \(\mathbf{P}_{W \mid Z, a}^\dagger\), we have that

\begin{equation}
\mathbf{P}_{h_1(Y) \mid Z, a'} = \mathbf{P}_{h_1(Y) \mid Z, a'} \mathbf{P}_{W \mid Z, a}^\dagger \mathbf{P}_{W \mid Z, a}
\end{equation}

Hence, \(\mathbf{P}_{h_1(Y) \mid Z, a'} \mathbf{P}_{W \mid Z, a}^\dagger \) defines \(h_0(y, w,a',a)\) in Assumption~\ref{assump:bridge_var1} such that

\[\sum_{w,m} h_1(y,\vec{m},a',w)\, p(w,\vec{m} \mid z,a) \\
 =
    \sum_{w} h_0(y, w,a',a)\, p(w \mid z,a).\]

\section{Identification Proof for Assumption Set 2}\label{app:id_set2}
\begin{proof}
Assumption~\ref{assump:ci_var2} and the first bridge solution in Assumption~\ref{assump:bridge_var2} gives 

\begin{equation}
\sum_{\vec{u}}p(y  \mid  u,a,\vec{m},\cancel{z})p(\vec{u}  \mid  a,\vec{m},z) = \sum_{w, \vec{u}} b_1(y,a,\vec{m},w)p\vec({w}  \mid  \vec{u}, a, \cancel{\vec{m},z})p(\vec{u}  \mid  a,\vec{m},z).
\end{equation}

Then, the first completeness condition in Assumption~\ref{assump:completeness_var2} gives

\begin{equation}\label{eq:bridge1_id_set2}
p(y  \mid  \vec{u},a,\vec{m}) = \sum_{w} b_1(y,a,\vec{m},w)p(w  \mid  \vec{u}, a).
\end{equation}

Similarly, Assumption~\ref{assump:ci_var2} and the second bridge solution in Assumption~\ref{assump:bridge_var2} gives 

\begin{equation}
\sum_{\vec{u}}p(\vec{m}  \mid  \vec{u},a,\cancel{w})p(\vec{u}  \mid  w, a) = \sum_{z, w}b_0(\vec{m},a,z)p(z  \mid  \vec{u},\cancel{w,a})p(\vec{u}  \mid  w, a).
\end{equation}

Then, the second completeness condition in Assumption~\ref{assump:completeness_var2} gives

\begin{equation}\label{eq:bridge2_id_set2}
p(\vec{m}  \mid  \vec{u},a) = \sum_{z} b_0(\vec{m}, a, z)p(z  \mid  \vec{u}).
\end{equation}

Hence, we have that
\begin{align}
    p(Y(a)=y) &= \sum_{\vec{u},\vec{m},{a'}} p(y \mid \vec{u},a',\vec{m})p(a' \mid \vec{u})p(\vec{m} \mid a,\vec{u})p(\vec{u})\\
    &= \sum_{\vec{u},\vec{m},a', w,z} b_1(\cdot) p(w  \mid  \vec{u}, a')p(a' \mid \vec{u})b_0(\cdot)p(z  \mid  \vec{u})p(\vec{u})  \text{ by Equation~\ref{eq:bridge1_id_set2} and Assumption~\ref{assump:ci_var2}}\\
    &= \sum_{\vec{u},\vec{m},a', w,z} b_1(\cdot) p(w, a'  \mid  \vec{u})b_0(\cdot)p(z  \mid  \vec{u})p(\vec{u})  \text{ by Equation~\ref{eq:bridge2_id_set2}}\\
    &= \sum_{\vec{u},\vec{m},a', w,z} b_1(\cdot) p(w, z, a' \mid  \vec{u})b_0(\cdot)p(\vec{u}) \text{ by Assumption~\ref{assump:ci_var2}}\\
    &=\sum_{\vec{m},a',w,z} b_1(y, a',\vec{m},w)b_0(\vec{m},a,z)p(w,z,a')
\end{align}

\end{proof}

\section{Bridge Solutions for Assumption~\ref{assump:bridge_var2} in Finite-Support Models}\label{app:bridge_set2}

For each \((a, \vec{m})\), let \[\mathbf{P}_{Y \mid Z, a,\vec{m}} = \begin{bmatrix}
\Pr(y_1 \mid z_1, a,\vec{m}) & \cdots & \Pr(y_1 \mid z_{|\mathcal{Z}|}, a,\vec{m}) \\ 
\Pr(y_2 \mid z_1, a,\vec{m}) & \cdots & \Pr(y_2 \mid z_{|\mathcal{Z}|}, a,\vec{m}) \\ 
\vdots & \ddots & \vdots \\ 
\Pr(y_{|\mathcal{Y}|} \mid z_1, a,\vec{m}) & \cdots & \Pr(y_{|\mathcal{Y}|} \mid z_{|\mathcal{Z}|}, a,\vec{m}) 
\end{bmatrix}\] 

and \[\mathbf{P}_{W \mid Z, a,\vec{m}} = \begin{bmatrix}
\Pr(w_1 \mid z_1, a,\vec{m}) & \cdots & \Pr(w_1 \mid z_{|\mathcal{Z}|}, a,\vec{m}) \\ 
\Pr(w_2 \mid z_1, a,\vec{m}) & \cdots & \Pr(w_2 \mid z_{|\mathcal{Z}|}, a,\vec{m}) \\ 
\vdots & \ddots & \vdots \\ 
\Pr(w_{|\mathcal{W}|} \mid z_1, a,\vec{m}) & \cdots & \Pr(w_{|\mathcal{W}|} \mid z_{|\mathcal{Z}|}, a,\vec{m}) 
\end{bmatrix}.\]

Then if distributions in \(\{p(W \mid z,a,\vec{m}): z \in \mathcal{Z}\)\} vary sufficiently across values \(z\) in its support \(\mathcal{Z}\) for each \((a,\vec{m})\) so that \(\mathbf{P}_{W \mid Z, a,\vec{m}}\) is left invertible with left inverse \(\mathbf{P}_{W \mid Z, a,\vec{m}}^\dagger\), we have that

\begin{equation}
\mathbf{P}_{Y \mid Z, a,\vec{m}} = \mathbf{P}_{Y \mid Z, a,\vec{m}} \mathbf{P}_{W \mid Z, a,\vec{m}}^\dagger \mathbf{P}_{W \mid Z, a,\vec{m}}
\end{equation}

Hence, \( \mathbf{P}_{Y \mid Z, a,\vec{m}} \mathbf{P}_{W \mid Z, a,\vec{m}}^\dagger  \) defines \(b_1(y,a,\vec{m},w)\) in Assumption~\ref{assump:bridge_var2} such that

\[p(y \mid z,a,\vec{m}) = \sum_{w} b_1(y,a,\vec{m},w)p(w \mid z,a,\vec{m}).\]

Similarly, for each \((a, \vec{m})\), define \(\mathbf{P}_{\vec{M} \mid W, a,\vec{m}}\)
and \(\mathbf{P}_{Z \mid W, a,\vec{m}}\).

Then if \(\mathbf{P}_{Z \mid W, a,\vec{m}}\) is left invertible with left inverse \(\mathbf{P}_{Z \mid W, a,]\vec{m}}^\dagger\), we have that

\begin{equation}
\mathbf{P}_{\vec{M} \mid W, a,\vec{m}} = \mathbf{P}_{\vec{M} \mid W, a,\vec{m}} \mathbf{P}_{Z \mid W, a,\vec{m}}^\dagger \mathbf{P}_{Z \mid W, a,\vec{m}}
\end{equation}

Hence, \(\mathbf{P}_{\vec{M} \mid W, a,\vec{m}} \mathbf{P}_{Z \mid W, a,\vec{m}}^\dagger \) defines \(b_0(\vec{m},a,z)\) in Assumption~\ref{assump:bridge_var2} such that

\[p(\vec{m} \mid w,a) = \sum_{z} b_0(\vec{m},a,z)p(z \mid w,a).\]

\section{Identification Proof for Assumption Set 3}\label{app:id_set3}
\begin{proof}
By Theorem 1 in \cite{hu2008instrumental}, under Assumptions~\ref{assump:ci_kp}-~\ref{assump:distinctness_var3}, \(p(Y, \vec{u}_i, W, Z \mid a,\vec m)\) is identified up to label \(i\) for each \((a, \vec m)\). Then, by \(W \ci A,\vec M \mid \vec U\) in Assumption~\ref{assump:ci_kp},
\(
p(W \mid u_i,a,\vec m)=p(W \mid u_i).
\)
Consequently, the label \(u_i\) can be matched across the conditional distributions
\(
p(Y,\vec u_i,W,Z \mid a,\vec m)
\)
for different values of \((a,\vec m)\), thereby recovering the full-data law
\(
p(A,\vec M,Y,\vec u_i,W,Z)\) up to label \(i\). Hence, \(p(Y(a)=y) = \sum_{\vec{u},\vec{m},a'} p(y \mid \vec{u}_i,a',\vec{m})p(a' \mid \vec{u}_i)p(\vec{m} \mid a,\vec{u}_i)p(\vec{u}_i)\) is identified in Fig.~\ref{fig:primal_v3}.
\end{proof}

\section{Finite-Support Identification Example under Assumption Set 3}\label{app:id_set3_discrete}

\begin{proof}
Suppose \(A, \vec{M}, Y, \vec{U}, W, Z\) are have finite support, and furthermore that \(W\), \(Z\), and \(\vec{U}\) have the same number of categories. 

Define 
\[
\mathbf{P}_{W \mid Z}
=
\begin{bmatrix}
\Pr(w_1 \mid z_1) & \cdots & \Pr(w_1 \mid z_n) \\
\Pr(w_2 \mid z_1) & \cdots & \Pr(w_2 \mid z_n) \\
\vdots & \ddots & \vdots \\
\Pr(w_n \mid z_1) & \cdots & \Pr(w_n \mid z_n)
\end{bmatrix},
\]
and analogously define $\mathbf{P}_{W \mid U}$, $\mathbf{P}_{U \mid Z}$, and $\mathbf{P}_{y,W \mid Z,a,\vec{m}}$. In other words, the $(k,j)^\text{th}$ entry of $\mathbf{P}_{y, W \mid Z,a,\vec{m}}$ is
\[
\mathbf{P}_{y, W \mid Z,a,\vec{m}}(k,j)
=
\Pr(Y=y, W=w_k \mid Z=z_j,A=a, \vec M=\vec{m}).
\]

Additionally, define \(\operatorname{diag}(\mathbf{P}_{y \mid U,a,\vec{m}})\) as the diagonal matrix with \(i^\text{th}\) diagonal entry \(
\Pr(Y=y\mid \vec{U}=\vec{u}_i,A=a, \vec M=\vec{m})
\).

By Assumption~\ref{assump:ci_kp}, 
\begin{equation}
    \mathbf{P}_{W \mid Z} = \mathbf{P}_{W \mid U} \mathbf{P}_{U \mid Z}, \text{ and} \\
\end{equation}
\begin{equation}
    \mathbf{P}_{y,W \mid Z,a,\vec{m}} = \mathbf{P}_{W \mid U} \operatorname{diag}(\mathbf{P}_{y \mid U,a,\vec{m}}) \mathbf{P}_{U \mid Z}.
\end{equation}

By Assumption~\ref{assump:completeness_var3}, \(\mathbf{P}_{U \mid Z}\) and \(\mathbf{P}_{W \mid U}\) are invertible. Then we have the following:
\[
\mathbf{P}_{W \mid Z}^{-1} = \mathbf{P}_{U \mid Z}^{-1}\, \mathbf{P}_{W \mid U}^{-1}.
\] Hence,
\begin{equation}
\mathbf{P}_{y,W \mid Z,a,\vec{m}} \mathbf{P}_{W \mid Z}^{-1} = \mathbf{P}_{W \mid U} \operatorname{diag}(\mathbf{P}_{y \mid U,a,\vec{m}})  \mathbf{P}_{W \mid U}^{-1}.
\end{equation}

Eigendecomposition of \(\mathbf{P}_{y,W \mid Z,a,\vec{m}} \mathbf{P}_{W \mid Z}^{-1}\) gives us the column vectors of \(\mathbf{P}_{W \mid U}\) under Assumption~\ref{assump:distinctness_var3}. In other words, we recover conditional distributions \(p(W \mid \vec{u}_i)\) [up to label \(i\)].  
Inverting   
\(
p(Z,A, \vec{M}, Y, W ) = \sum_{i} p(Z,A, \vec{M}, Y, \vec{u}_i)\, p(W \mid \vec{u}_i),
\)  
which holds by Assumption~\ref{assump:ci_kp}, we recover \(p(Z,A,\vec{M},Y \mid \vec{u}_i)\) and \(p(\vec{u}_i)\).

Hence, we recover \(p(A,\vec{M},Y,\vec{u}_i,W,Z) = p(W \mid \vec{u}_i) p(Z,A,\vec{M},Y \mid \vec{u}_i) p(\vec{u}_i) \).

This proof extends to settings in which \(W\) and \(Z\) have higher cardinality than \(\vec{U}\) by appropriately coarsening variables. 
\end{proof}

\section{Plug-In Estimators for Assumption Sets 1, 2, and 3}\label{app:plug_in_estimation}

\subsection{Plug-In Estimator for Assumption Set 1}

Under Assumptions~\ref{assump:ci_var1}--\ref{assump:bridge_var1}, the counterfactual distribution in Figure~\ref{fig:primal_v1} is identified as in Theorem~\ref{thm:id_set1}.

Define the mean-induced bridge:
\begin{equation}\label{eq:mean_bridge1}
\bar h_0(w,a',a)
\;:=\;
\sum_{y} y\, h_0(y,w,a',a).
\end{equation}

Then the counterfactual mean satisfies
\begin{equation}
\mathbb{E}\!\left[Y(a)\right]
\;=\;
\sum_{a',w} \bar h_0(w,a',a)\, p(w,a')
\;=\;
\mathbb{E}\!\left[\bar h_0(W,A',a)\right],
\end{equation}
and hence the average causal effect can be written as
\begin{equation}
\mathbb{E}\!\left[Y(1)-Y(0)\right]
\;=\;
\mathbb{E}\!\left[\bar h_0(W,A',1)\right]
-
\mathbb{E}\!\left[\bar h_0(W,A',0)\right].
\end{equation}

We estimate this quantity using the plug-in estimator
\begin{equation}\label{eq:S1_PI}
\hat\psi^{\text{(S1,PI)}}
=
\mathbb{E}_n\!\left[\bar h_0(W,A',1)\right]
-
\mathbb{E}_n\!\left[\bar h_0(W,A',0)\right].
\end{equation}
Provided that $\bar h_0(W,A',a_1) \xrightarrow{p} \bar h_0(W,A',a;\eta_1)$ and standard  regularity conditions hold, consistency of $\hat\psi^{\text{(S1)}}$ follows.

\subsection{Plug-In Estimator for Assumption Set 2}
Under Assumptions~\ref{assump:ci_var2}--\ref{assump:bridge_var2}, the counterfactual distribution in Figure~\ref{fig:primal_v2} is identified as in Theorem~\ref{thm:id_set2}.

Define the mean-induced bridge:
\begin{equation}\label{eq:mean_bridge2}
\bar b_2(a',a,w,z)
\;:=\;
\sum_{y,\vec{m}} y\, b_1(y,a',\vec{m},w) b_0(\vec{m},a,z).
\end{equation}

Then the counterfactual mean satisfies
\begin{equation}
\mathbb{E}\!\left[Y(a)\right]
\;=\;
\sum_{a',w,z}
\bar b_2(a',a,w,z)\,
p(w,z,a'),
\end{equation}
and hence the average causal effect can be written as
\begin{equation}
\mathbb{E}\!\left[Y(1)-Y(0)\right]
=
\mathbb{E}\!\left[
\bar b_2(A',1,W,Z)
\right] -
\mathbb{E}\!\left[
\bar b_2(A',0,W,Z)
\right].
\end{equation}

We estimate this quantity using the estimator
\begin{equation}\label{eq:S2_PI}
\hat\psi^{\text{(S2,PI)}} =
\mathbb{E}_n\!\left[
\bar b_2(A',1,W,Z_2)
\right] -
\mathbb{E}_n\!\left[
\bar b_2(A',0,W,Z_2)
\right].
\end{equation}

Provided that \(\bar b_2(A',1,W,Z_2)\xrightarrow{p}\bar b_2(A',0,W,Z;\eta_2)\) and standard regularity conditions hold, consistency of $\hat\psi^{\text{(S2)}}$ follows.

\subsection{Plug-In Estimator for Assumption Set 3}
By Theorem \ref{thm:id_set3}, under Assumptions~\ref{assump:ci_kp}-\ref{assump:distinctness_var3}, the counterfactual distribution in Fig. \ref{fig:primal_v3} can be derived from the full law \[\{ p(A,\vec{M},Y,\vec{u}_i,W,Z): \vec{u}_i \} \]  up to relabeling of $u_i$.  Distributional components $p(Y|\vec{m},a',\vec{u}_i)$, $p(a'|\vec{u}_i)$, and $p(\vec{m}|a,\vec{u}_i)$ can be estimated for each level of $\vec{u}_i$, e.g. via eigendecomposition tasks details in Appendix~\ref{app:id_set3_discrete}. Then, we construct the plug-in estimator 
\begin{equation}\label{eq:S3_PI}
 \hat{\psi}^{(S3,PI)} 
 = \sum_{\vec{u}_i,\vec{m},a'} 
 \hat E[Y \mid \vec{m},a',\vec{u}_i] \,
 \hat p(a' \mid \vec{u}_i)
 [\hat p(\vec{m} \mid a=1,\vec{u}_i)-\hat p(\vec{m} \mid a=0,\vec{u}_i)]
 \hat p(\vec{u}_i).
\end{equation}

\section{Proxy Weight Constructions for finite-support hidden variable $U$}\label{app:proxy_contrasts}

For each \((\vec{u}_i,a)\), define
\begin{equation}
\begin{aligned}
&\mathcal{C}_{WZ}(\vec{u}_i,a) =\\
& \prod_{j \neq i} \Big\{
\frac{h_j(W) - \mathbb{E}\!\left[h_j(W)\mid \vec{U}=\vec{u}_j\right]}
     {\mathbb{E}\!\left[h_j(W)\mid \vec{U}=\vec{u}_i\right]
      - \mathbb{E}\!\left[h_j(W)\mid \vec{U}=\vec{u}_j\right]} \cdot
\frac{f_j(Z) - \mathbb{E}\!\left[f_j(Z)\mid \vec{U}=\vec{u}_j,a\right]}
     {\mathbb{E}\!\left[f_j(Z)\mid \vec{U}=\vec{u}_i,a\right]
      - \mathbb{E}\!\left[f_j(Z)\mid \vec{U}=\vec{u}_j,a\right]} \Big\}, 
\end{aligned}
\end{equation}

and for each \((\vec{u}_i,a,\vec{m})\), define

\begin{equation}
\begin{aligned}
&\mathcal{C}_{YW}(\vec{u}_i,a,\vec{m}) =\\
& \prod_{j \neq i} \Big\{
\frac{t_j(Y) - \mathbb{E}\!\left[t_j(Y)\mid \vec{U}=\vec{u}_j,a,\vec{m}\right]}
     {\mathbb{E}\!\left[t_j(Y)\mid \vec{U}=\vec{u}_i,a,\vec{m}\right]
      - \mathbb{E}\!\left[t_j(Y)\mid \vec{U}=\vec{u}_j,a,\vec{m}\right]} \cdot
\frac{h_j(W) - \mathbb{E}\!\left[h_j(W)\mid \vec{U}=\vec{u}_j\right]}
     {\mathbb{E}\!\left[h_j(W)\mid \vec{U}=\vec{u}_i\right]
      - \mathbb{E}\!\left[h_j(W)\mid \vec{U}=\vec{u}_j\right]} \Big\},
\end{aligned}
\end{equation}

\begin{equation}
\begin{aligned}
&\mathcal{C}_{YZ}(\vec{u}_i,a,\vec{m}) =\\
& \prod_{j \neq i} \Big\{
\frac{t_j(Y) - \mathbb{E}\!\left[t_j(Y)\mid \vec{U}=\vec{u}_j,a,\vec{m}\right]}
     {\mathbb{E}\!\left[t_j(Y)\mid \vec{U}=\vec{u}_i,a,\vec{m}\right]
      - \mathbb{E}\!\left[t_j(Y)\mid \vec{U}=\vec{u}_j,a,\vec{m}\right]} \cdot
\frac{f_j(Z) - \mathbb{E}\!\left[f_j(Z)\mid \vec{U}=\vec{u}_j,a\right]}
     {\mathbb{E}\!\left[f_j(Z)\mid \vec{U}=\vec{u}_i,a\right]
      - \mathbb{E}\!\left[f_j(Z)\mid \vec{U}=\vec{u}_j.a\right]} \Big\},
\end{aligned}
\end{equation}

\begin{equation}
\begin{aligned}
&\mathcal{C}_{YWZ}(\vec{u}_i,a,\vec{m}) =\\
&\prod_{j \neq i} \Big\{
\frac{t_j(Y) - \mathbb{E}\!\left[t_j(Y)\mid \vec{U}=\vec{u}_j,a,\vec{m}\right]}
     {\mathbb{E}\!\left[t_j(Y)\mid \vec{U}=\vec{u}_i,a,\vec{m}\right]
      - \mathbb{E}\!\left[t_j(Y)\mid \vec{U}=\vec{u}_j,a,\vec{m}\right]} \cdot
\frac{h_j(W) - \mathbb{E}\!\left[h_j(W)\mid \vec{U}=\vec{u}_j\right]}
     {\mathbb{E}\!\left[h_j(W)\mid \vec{U}=\vec{u}_i\right]
      - \mathbb{E}\!\left[h_j(W)\mid \vec{U}=\vec{u}_j\right]}\\
& \qquad \cdot
\frac{f_j(Z) - \mathbb{E}\!\left[f_j(Z)\mid \vec{U}=\vec{u}_j,a\right]}
     {\mathbb{E}\!\left[f_j(Z)\mid \vec{U}=\vec{u}_i,a\right]
      - \mathbb{E}\!\left[f_j(Z)\mid \vec{U}=\vec{u}_j,a\right]} \Big\},
\end{aligned}
\end{equation}
where the collections of random functions 
$\{h_j(W)\}_j$, $\{f_j(Z)\}_j$, and $\{t_j(Y)\}_j$ are chosen such that for \(j' \neq j \neq i\), 
\begin{equation}
\begin{aligned}
&\mathbb E\!\left[
\prod_{j\neq i} h_j(W)
\,\middle|\, \vec U=\vec u_i
\right] =
\prod_{j\neq i}
\mathbb E\!\left[
h_j(W)
\,\middle|\, \vec U=\vec u_i
\right];
\end{aligned}
\end{equation}
\begin{equation}
\begin{aligned}
&\mathbb E\!\left[
\prod_{j\neq i} f_j(Z)
\,\middle|\, \vec U=\vec u_i,a
\right] =
\prod_{j\neq i}
\mathbb E\!\left[
f_j(Z)
\,\middle|\, \vec U=\vec u_i,a
\right];
\end{aligned}
\end{equation}

\begin{equation}
\begin{aligned}
&\mathbb E\!\left[
\prod_{j\neq i} t_j(Y)
\,\middle|\, \vec U=\vec u_i,a,\vec \vec m
\right] =
\prod_{j\neq i}
\mathbb E\!\left[
t_j(Y)
\,\middle|\, \vec U=\vec u_i,a,\vec \vec m
\right].
\end{aligned}
\end{equation}

We include the identity functions on $W$, $Z$, and $Y$ in the collections $\{h_j(W)\}_j$, $\{f_j(Z)\}_j$, and $\{t_j(Y)\}_j$, respectively. When $\vec U$ is binary, these collections reduce to the identity functions on $W$, $Z$, and $Y$.

\section{Influence function for \(\psi_a\) in Figure~\ref{fig:primal_v3} without $\vec{U}$ hidden}\label{app:full_if}
\begin{proof}
We aim to show \(\varphi_{\mathrm{full}}(A,\vec{M},Y,\vec{U};\cdot)\) from Equation~\ref{eq:full_eif} is an influence function for \(\psi_a\) in Figure~\ref{fig:primal_v3} without $\vec{U}$ hidden (variables \(A,\vec{M}, Y, \vec{U}, W,Z\) are  all observed). In other words, that
\[
\left.\frac{d}{dt}\psi(P_\epsilon)\right|_{\epsilon=0}
=
\mathbb{E}_{A,\vec{M},Y,\vec{U},W,Z}\!\left[
\varphi_{\mathrm{full}}(A,\vec{M},Y,\vec{U};\cdot)\,
s_\epsilon(A,\vec{M},Y,\vec{U},W,Z)
\right],\]
where score for regular parametric submodel $p_\epsilon(A,\vec{M},Y,\vec{U},W,Z)$ is defined as \[
s_\epsilon(A,\vec{M},Y,\vec{U},W,Z)=\left.\frac{d}{d\epsilon}\log p_\epsilon(A,\vec{M},Y,\vec{U},W,Z)\right|_{\epsilon=0}
\] i.e., \[p_\epsilon(A,\vec{M},Y,\vec{U},W,Z) 
= p(A,\vec{M},Y,\vec{U},W,Z) \big(1 + \epsilon s_\epsilon(A,\vec{M},Y,\vec{U},W,Z)\big) + o(\epsilon).\]

We can construct a regular parametric submodel for the marginal law through marginalization of a regular parametric submodel for the full data law.
\begin{align*}
p_\epsilon(A,\vec{M},Y,\vec{U}) 
&= \sum_{w,z} [p(A, \vec{M},Y,\vec{U},w,z) \big(1 + \epsilon s_\epsilon(A,\vec{M},Y,\vec{U},w,z)\big) + o(\epsilon)] \\
&= p(A,\vec{M},Y,\vec{U}) \left[ 1 + \epsilon \, \mathbb{E}[s_\epsilon(A,\vec{M},Y,\vec{U},W,Z) \mid A,\vec{M},Y,\vec{U}] \right] + o(\epsilon)\\
&= p(A,\vec{M},Y,\vec{U}) \big(1 + \epsilon s_\epsilon(A,\vec{M},Y,\vec{U})\big) + o(\epsilon),
\end{align*}

i.e., \[
s_\epsilon(A,\vec{M},Y,\vec{U})=\left.\frac{d}{d\epsilon}\log p_\epsilon(A,\vec{M},Y,\vec{U})\right|_{\epsilon=0}.
\]

\citet{fulcher2020robust} already showed that
\[
\left.\frac{d}{dt}\psi(P_{\epsilon})\right|_{\epsilon=0}
=
\mathbb{E}_{A,\vec{M},Y,\vec{U}}\!\left[
\varphi_{\mathrm{full}}(A,\vec{M},Y,\vec{U};\cdot)\,
s_{\epsilon}(A,\vec{M},Y,\vec{U})
\right].\]

Then by Assumption~\ref{assump:ci_kp}, which also implies the factorization
\[
p(A,\vec{M},Y,\vec{U},W,Z)
=
p(W \mid \vec{M},A,\vec{U})\, p(Z \mid \vec{M},A,\vec{U})\, p(A,\vec{M},Y,\vec{U}),
\]
we have
\[
s_{\epsilon}(A,\vec{M},Y,\vec{U},W,Z)
=s_{\epsilon}(A,\vec{M},Y,\vec{U}) + s_{\epsilon}(W \mid \vec{M},A,\vec{U}) + s_{\epsilon}(Z \mid \vec{M},A,\vec{U}).
\]
Finally by iterated expectations,
\begin{align*}
&\left.\frac{d}{dt}\psi(P_{\epsilon})\right|_{\epsilon=0} =
\mathbb{E}_{A,\vec{M},Y,\vec{U}}\!\left[
\varphi_{\mathrm{full}}(A,\vec{M},Y,\vec{U};\cdot)\,
s_{\epsilon}(A,\vec{M},Y,\vec{U})\right] + 0 + 0 \\
& \qquad \qquad \qquad= \mathbb{E}_{A,\vec{M},Y,\vec{U}}\!\left[
\varphi_{\mathrm{full}}(A,\vec{M},Y,\vec{U};\cdot)\,
s_{\epsilon}(A,\vec{M},Y,\vec{U})\right] \\
& \qquad \qquad \qquad+ \mathbb{E}_{A,\vec{M},Y,\vec{U}}\!\left[
\varphi_{\mathrm{full}}(A,\vec{M},Y,\vec{U};\cdot)\,
\mathbb{E}[s_{\epsilon}(W \mid \vec{M},A,\vec{U}) \mid A, \vec{M}, Y, \vec{U}]\right] \\
& \qquad \qquad \qquad + \mathbb{E}_{A,\vec{M},Y,\vec{U}}\!\left[
\varphi_{\mathrm{full}}(A,\vec{M},Y,\vec{U};\cdot)\,
\mathbb{E}[s_{\epsilon}(Z \mid \vec{M},A,\vec{U}) \mid A, \vec{M}, Y, \vec{U}]\right]\\
& \qquad \qquad \qquad=
\mathbb{E}_{A,\vec{M},Y,\vec{U},W,Z}\!\left[
\varphi_{\mathrm{full}}(A,\vec{M},Y,\vec{U};\cdot)\,
s_{\epsilon}(A,\vec{M},Y,\vec{U},W,Z)
\right].
\end{align*}
\end{proof}

\section{Proof of Theorem~\ref{thm:obs_if}}\label{app:obs_if}
\begin{proof}
A regular parametric submodel for the observed data law is induced by a regular parametric submodel for the full data law through marginalization. 
\begin{align*}
p_\epsilon(A,\vec{M},Y,W,Z) 
&= \sum_up_\epsilon(A,\vec{M},Y,\vec{U},W,Z)  \\
&= \sum_u [p(A, \vec{M},Y,u,W,Z) \big(1 + \epsilon s_\epsilon(A,\vec{M},Y,u,W,Z)\big) + o(\epsilon)] \\
&= p(A,\vec{M},Y,W,Z) \left[ 1 + \epsilon \, \mathbb{E}[s_\epsilon(A,\vec{M},Y,\vec{U},W,Z) \mid A,\vec{M},Y,W,Z] \right] + o(\epsilon)\\
&= p(A,\vec{M},Y,W,Z) \big(1 + \epsilon s_\epsilon(A,\vec{M},Y,W,Z)\big) + o(\epsilon),
\end{align*}
i.e.,
\[
s_\epsilon(A,\vec{M},Y,W,Z)=\left.\frac{d}{d\epsilon}\log p_\epsilon(A,\vec{M},Y,W,Z)\right|_{\epsilon=0}.
\] 

We aim to show \(\varphi_{\mathrm{obs}}(A,\vec{M},Y,W,Z;\cdot)\) from Equation~\ref{eq:obs_if} is an influence function for \(\psi_a\) in Figure~\ref{fig:primal_v3} with $\vec{U}$ hidden ($A,\vec{M}, Y, W, Z$ observed):
\[
\left.\frac{d}{dt}\psi(P_\epsilon)\right|_{\epsilon=0}
=
\mathbb{E}_{A,\vec{M},Y,W,Z}\!\left[
\varphi_{\mathrm{obs}}(A,\vec{M},Y,W,Z;\cdot)\,
s_\epsilon(A,\vec{M},Y,W,Z)
\right].\]

We already know from Appendix~\ref{app:full_if} that
\[
\left.\frac{d}{dt}\psi(P_{\epsilon})\right|_{\epsilon=0}
=
\mathbb{E}_{A,\vec{M},Y,\vec{U},W,Z}\!\left[
\varphi_{\mathrm{full}}(A,\vec{M},Y,\vec{U};\cdot)\,
s_{\epsilon}(A,\vec{M},Y,\vec{U},W,Z)
\right].
\]

Breaking up \(s_{\epsilon}(A,\vec{M},Y,\vec{U},W,Z)\),
\begin{align*}
\left.\frac{d}{dt}\psi(P_{\epsilon})\right|_{\epsilon=0}
&=
\mathbb{E}\!\left[
\varphi_{\mathrm{full}}(A,\vec{M},Y,\vec{U};\cdot)\,
s_{\epsilon}(A,\vec{M},Y,\vec{U})
\right]  \\
&\quad +
\mathbb{E}\!\left[
\varphi_{\mathrm{full}}(A,\vec{M},Y,\vec{U};\cdot)\,
s_{\epsilon}(W \mid \vec{M},A,\vec{U})
\right] \\
&\quad +
\mathbb{E}\!\left[
\varphi_{\mathrm{full}}(A,\vec{M},Y,U;\cdot)\,
s_{\epsilon}(Z \mid \vec{M},A,\vec{U})
\right] \\[6pt]
&=
\mathbb{E}\!\left[
\mathbb{E}\!\left\{
\varphi_{\mathrm{full}}(A,\vec{M},Y,\vec{U};\cdot)
\,\middle|\, A,\vec{M},Y,\vec{U}
\right\}
\, s_{\epsilon}(A,\vec{M},Y,\vec{U})
\right] \\
&\quad +
\mathbb{E}\!\left[
\mathbb{E}\!\left\{
\varphi_{\mathrm{full}}(A,\vec{M},Y,\vec{U};\cdot)
\,\middle|\, \vec{M},A,\vec{U}, W
\right\}
\, s_{\epsilon}(W \mid \vec{M},A,\vec{U})
\right] \\
&\quad +
\mathbb{E}\!\left[
\mathbb{E}\!\left\{
\varphi_{\mathrm{full}}(A,\vec{M},Y,\vec{U};\cdot)
\,\middle|\, \vec{M},A,\vec{U}, Z
\right\}
\, s_{\epsilon}(Z \mid \vec{M},A,\vec{U})
\right].
\end{align*}

Hence, it suffices to verify that the following conditional expectations coincide:
\begin{enumerate}
\item
\(
\mathbb{E}\!\left[
\varphi_{\mathrm{obs}}(A,\vec{M},Y,W,Z;\cdot)
\,\middle|\, A,\vec{M},Y,\vec{U}
\right]
=
\mathbb{E}\!\left[
\varphi_{\mathrm{full}}(A,\vec{M},Y,\vec{U};\cdot)
\,\middle|\, A,\vec{M},Y,\vec{U}
\right]
\)

\item
\(
\mathbb{E}\!\left[
\varphi_{\mathrm{obs}}(A,\vec{M},Y,W,Z;\cdot)
\,\middle|\, \vec{M},A,\vec{U}, W
\right]
=
\mathbb{E}\!\left[
\varphi_{\mathrm{full}}(A,\vec{M},Y,\vec{U};\cdot)
\,\middle|\, \vec{M},A,\vec{U}, W
\right]
\)

\item
\(
\mathbb{E}\!\left[
\varphi_{\mathrm{obs}}(A,\vec{M},Y,W,Z;\cdot)
\,\middle|\, \vec{M},A,\vec{U},Z
\right]
=
\mathbb{E}\!\left[
\varphi_{\mathrm{full}}(A,\vec{M},Y,\vec{U};\cdot)
\,\middle|\, \vec{M},A,\vec{U}, Z
\right]
\)
\end{enumerate}

because then
\begin{align*}
&\left.\frac{d}{dt}\psi(P_{\epsilon})\right|_{\epsilon=0}
=
\mathbb{E}_{A,\vec{M},Y,\vec{U},W,Z}\!\left[
\varphi_{\mathrm{obs}}(A,\vec{M},Y,W,Z;\cdot)\,
s_{\epsilon}(A,\vec{M},Y,\vec{U},W,Z)
\right] \\
& \qquad \qquad \qquad =
\mathbb{E}_{A,\vec{M},Y,W,Z}\!\left[
\varphi_{\mathrm{obs}}(A,\vec{M},Y,W,Z;\cdot)\,
\mathbb{E}[s_{\epsilon}(A,\vec{M},Y,\vec{U},W,Z)\mid A,\vec{M},Y,W,Z]\right]  \\
& \qquad \qquad \qquad =
\mathbb{E}_{A,\vec{M},Y,W,Z}\!\left[
\varphi_{\mathrm{obs}}(A,\vec{M},Y,W,Z;\cdot)\,
s_{\epsilon}(A,\vec{M},Y,W,Z)\right].
\end{align*}

Indeed, (1)-(3) hold noting Equation~\ref{eqn:weight_prop1} and~\ref{eqn:weight_prop2}, as well as Assumption~\ref{assump:ci_kp}:

\begin{enumerate}
\item
\(
\mathbb{E}\!\left[
\varphi_{\mathrm{obs}}(A,\vec{M},Y,W,Z;\cdot)
\,\middle|\, A,\vec{M},Y,\vec{U}
\right]
=
\mathbb{E}\!\left[
\varphi_{\mathrm{full}}(A,\vec{M},Y,\vec{U};\cdot)
\,\middle|\, A,\vec{M},Y,\vec{U}
\right]
=\varphi_{\mathrm{full}}(A,\vec{M},Y,\vec{U};\cdot)
\)

\item
\(\mathbb{E}\!\left[
\varphi_{\mathrm{obs}}(A,\vec{M},Y,W,Z;\cdot)
\,\middle|\, \vec{M},A,\vec{U},W
\right] =
\mathbb{E}\!\left[
\varphi_{\mathrm{full}}(A,\vec{M},Y,\vec{U};\cdot)
\,\middle|\, \vec{M},A,\vec{U},W
\right] = \varphi_{\mathrm{full}_2}(\vec{M},A,\vec{U})) + 
\varphi_{\mathrm{full}_3}(A,\vec{U})  - \psi_{a}\)

\item
\(
\mathbb{E}\!\left[
\varphi_{\mathrm{obs}}(A,\vec{M},Y,W,Z;\cdot)
\,\middle|\, \vec{M},A,\vec{U},Z
\right]
=
\mathbb{E}\!\left[
\varphi_{\mathrm{full}}(A,\vec{M},Y,\vec{U};\cdot)
\,\middle|\, \vec{M},A,\vec{U},Z
\right] = \varphi_{\mathrm{full}_2}(\vec{M},A,\vec{U})) + 
\varphi_{\mathrm{full}_3}(A,\vec{U}) - \psi_{a}
\)
\end{enumerate}

\end{proof}

\section{Proof of Theorem~\ref{thm:multiple_robustness}}\label{app:multiple_robustness}
\begin{proof}
Solving Equation~\ref{eq:est_equation} with $\varphi_{obs}(A,\vec{M},Y,W,Z; \cdot)$ in Equation~\ref{eq:obs_if} gives

\begin{align*}
&\hat \psi_a^{(S3,IF)} = \mathbb{E}_n\big[\sum_{i}
\mathcal{C}_{WZ}(A,\vec{u}_i)\,
\varphi_{\mathrm{full}_1}(Y,\vec{M},A,\vec{u}_i) \big] \\
&\quad+
\mathbb{E}_n\!\Big[
\sum_i
\Big\{
\mathcal{C}_{WZ}(A,\vec{u}_i)
+
\mathcal{C}_{YW}(\vec{M},A,\vec{u}_i)
+
\mathcal{C}_{YZ}(\vec{M},A,\vec{u}_i)
-2\,\mathcal{C}_{WZY}(\vec{M},A,\vec{u}_i)
\Big\}  \\
&\qquad\qquad\qquad\qquad\cdot
\Big\{
\varphi_{\mathrm{full}_2}(\vec{M},A,\vec{u}_i)
+
\varphi_{\mathrm{full}_3}(A,\vec{u}_i)
\Big\}
\Big].
\end{align*}

In the following arguments, we use $\mathbb{E}_{n,\text{sub}}[\cdot \mid \cdot]$ to denote the empirical average taken only over the subsample of observations satisfying the conditioning event appearing after the vertical bar $\mid$. Standard regularity conditions are assumed.

\textbf{Case 1.} \(p(\vec{M} \mid A, \vec{U})\), \(\{\mathbb{E}[h_j(W) \mid \vec{U}]\}_j\), \(\{\mathbb{E}[f_j(Z) \mid \vec{U},A]\}_j\) are correctly specified.

Since $\{\mathbb{E}[h_j(W)\mid \vec{U}]\}_j$ and $\{\mathbb{E}[f_j(Z)\mid \vec{U},A]\}_j$ are correctly specified, and Assumption~\ref{assump:ci_kp} holds, 
\[
\mathbb{E}_{n, \text{sub}}\!\left[
\mathcal{C}_{WZ}(A,\vec{u}_i)
\,\middle|\,
Y,\vec{M},A,\vec{U}=\vec{u}_i
\right]
\xrightarrow{p}
\mathbb{E}\!\left[
\mathcal{C}_{WZ}(A,\vec{u}_i)
\,\middle|\,
Y,\vec{M},A,\vec{U}=\vec{u}_i
\right]
= 1, \text{ and}
\]
\[
\mathbb{E}_{n, \text{sub}}\!\left[
\mathcal{C}_{WZ}(A,\vec{u}_j)
\,\middle|\,
Y,\vec{M},A,\vec{U}=\vec{u}_i
\right]
\xrightarrow{p}
\mathbb{E}\!\left[
\mathcal{C}_{WZ}(A,\vec{u}_j)
\,\middle|\,
Y,\vec{M},A,\vec{U}=\vec{u}_i
\right]
= 0, j \neq i.
\]

Additionally, 

\begin{align*}
&\mathbb{E}_{n,\text{sub}}
[\mathcal{C}_{WZ}(A,\vec{u}_i)+ \mathcal{C}_{YW}(\vec{M},A,\vec{u}_i)+ \mathcal{C}_{YZ}(\vec{M},A,\vec{u}_i) -2\,\mathcal{C}_{WZY}(\vec{M},A,\vec{u}_i) \mid \vec{M},A,\vec{U} = \vec{u}_i] 
\xrightarrow{p} \\
&\mathbb{E}
[\mathcal{C}_{WZ}(A,\vec{u}_i)+ \mathcal{C}_{YW}(\vec{M},A,\vec{u}_i)+ \mathcal{C}_{YZ}(\vec{M},A,\vec{u}_i) -2\,\mathcal{C}_{WZY}(\vec{M},A,\vec{u}_i)\mid \vec{M},A,\vec{U}= \vec{u}_i] 
= 1, \text{ and}
\end{align*}
\begin{align*}
&\mathbb{E}_{n,\text{sub}}
[\mathcal{C}_{WZ}(A,\vec{u}_j)+ \mathcal{C}_{YW}(\vec{M},A,\vec{u}_j)+ \mathcal{C}_{YZ}(\vec{M},A,\vec{u}_j) -2\,\mathcal{C}_{WZY}(\vec{M},A,\vec{u}_j) \mid \vec{M},A,\vec{U} = \vec{u}_i] 
\xrightarrow{p} \\
&\mathbb{E}
[\mathcal{C}_{WZ}(A,\vec{u}_j)+ \mathcal{C}_{YW}(\vec{M},A,\vec{u}_j)+ \mathcal{C}_{YZ}(\vec{M},A,\vec{u}_j) -2\,\mathcal{C}_{WZY}(\vec{M},A,\vec{u}_j)\mid \vec{M},A,\vec{U}= \vec{u}_i] 
= 0, j \neq i.
\end{align*}

\citet{fulcher2020robust} show that, provided $p(\vec{M}\mid A,\vec{U})$ is correctly specified,
\[
\mathbb{E}_n\!\Big[
\varphi_{\mathrm{full}_1}(Y,\vec{M},A,\vec{U})
+
\varphi_{\mathrm{full}_2}(\vec{M},A,\vec{U})
+
\varphi_{\mathrm{full}_3}(A,\vec{U})
\Big]
\xrightarrow{p}
\psi_a .
\]

Combining these results, it follows that
\begin{align*}
\hat \psi_a^{(S3, IF)}
&=
\mathbb{E}_n\!\Big[
\mathbb{E}_{n,\mathrm{sub}}\!\Big[
\sum_i
\mathcal{C}_{WZ}(A,\vec{u}_i)\,
\varphi_{\mathrm{full}_1}(Y,\vec{M},A,\vec{u}_i)
\Big]
\ \Big|\ Y,\vec{M},A,\vec{U}
\Big] \\
&\quad+
\mathbb{E}_n\!\Big[
\mathbb{E}_{n,\mathrm{sub}}\!\Big[
\sum_i
\Big\{
\mathcal{C}_{WZ}(A,\vec{u}_i)
+
\mathcal{C}_{YW}(\vec{M},A,\vec{u}_i)
+
\mathcal{C}_{YZ}(\vec{M},A,\vec{u}_i)
-2\,\mathcal{C}_{WZY}(\vec{M},A,\vec{u}_i)
\Big\}  \\
&\qquad\qquad\qquad\qquad\cdot
\Big\{
\varphi_{\mathrm{full}_2}(\vec{M},A,\vec{u}_i)
+
\varphi_{\mathrm{full}_3}(A,\vec{u}_i)
\Big\}
\Big]
\ \Big|\ \vec{M},A,\vec{U}
\Big]
\xrightarrow{p}
\psi_a .
\end{align*}

\textbf{Case 2.} \(p(\vec{M} \mid A, \vec{U})\), \(\{\mathbb{E}[h_j(W) \mid \vec{U}]\}_j\), \(\{\mathbb{E}[t_j(Y) \mid \vec{M}, A, \vec{U}]\}_j\) are correctly specified.

Since \(\{\mathbb{E}[h_j(W) \mid \vec{U}]\}_j\) is correctly specified, and Assumption~\ref{assump:ci_kp} holds, 
\[
\mathbb{E}_{n,\text{sub}}\big[
\mathcal{C}_{WZ}(A,\vec{u}_i) \mid Y, \vec{M},A,\vec{U}=\vec{u}_i\big] \xrightarrow{p} k, \text{ where } k \text{ is a constant, and }
\]
\[
\mathbb{E}_{n,\text{sub}}\big[
\mathcal{C}_{WZ}(A,\vec{u}_j) \mid Y, \vec{M},A,\vec{U}=\vec{u}_i\big] \xrightarrow{p} 0, j \neq i.
\]

Then, since \(\{\mathbb{E}[t_j(Y) \mid \vec{M}, A, \vec{U}]\}_j\) (thereby $\mathbb{E}[Y \mid \vec{M},A,\vec{U}]$) is correctly specified,
\begin{align*}
\mathbb{E}_n\!\Big[
\mathbb{E}_{n,\mathrm{sub}}\!\Big[
\sum_i
\mathcal{C}_{WZ}(A,\vec{u}_i)\,
\varphi_{\mathrm{full}_1}(Y,\vec{M},A,\vec{u}_i)
\Big]
\ \Big|\ Y,\vec{M},A,\vec{U}
\Big] \xrightarrow{p} 0.
\end{align*}

Since \(\{\mathbb{E}[h_j(W) \mid \vec{U}]\}_j\), \(\{\mathbb{E}[t_j(Y) \mid \vec{M}, A, \vec{U}]\}_j\) are correctly specified, and Assumption~\ref{assump:ci_kp} holds,
\begin{align*}
&\mathbb{E}_n
[\mathcal{C}_{WZ}(A,\vec{u}_i)+ \mathcal{C}_{YW}(\vec{M},A,\vec{u}_i)+ \mathcal{C}_{YZ}(\vec{M},A,\vec{u}_i) -2\,\mathcal{C}_{WZY}(\vec{M},A,\vec{u}_i) \mid \vec{M},A,\vec{U} = \vec{u}_i] 
\xrightarrow{p} \\
&\mathbb{E}
[\mathcal{C}_{WZ}(A,\vec{u}_i)+ \mathcal{C}_{YW}(\vec{M},A,\vec{u}_i)+ \mathcal{C}_{YZ}(\vec{M},A,\vec{u}_i) -2\,\mathcal{C}_{WZY}(\vec{M},A,\vec{u}_i)\mid \vec{M},A,\vec{U}= \vec{u}_i] 
= 1, \text{ and }
\end{align*}
\begin{align*}
&\mathbb{E}_{n,\text{sub}}
[\mathcal{C}_{WZ}(A,\vec{u}_j)+ \mathcal{C}_{YW}(\vec{M},A,\vec{u}_j)+ \mathcal{C}_{YZ}(\vec{M},A,\vec{u}_j) -2\,\mathcal{C}_{WZY}(\vec{M},A,\vec{u}_j) \mid \vec{M},A,\vec{U} = \vec{u}_i] 
\xrightarrow{p} \\
&\mathbb{E}
[\mathcal{C}_{WZ}(A,\vec{u}_j)+ \mathcal{C}_{YW}(\vec{M},A,\vec{u}_j)+ \mathcal{C}_{YZ}(\vec{M},A,\vec{u}_j) -2\,\mathcal{C}_{WZY}(\vec{M},A,\vec{u}_j)\mid \vec{M},A,\vec{U}= \vec{u}_i] 
= 0, j \neq i.
\end{align*}

\citet{fulcher2020robust} show that, provided $p(\vec{M}\mid A,\vec{U})$ is correctly specified,
\[
\mathbb{E}_n\!\Big[
\varphi_{\mathrm{full}_2}(\vec{M},A,\vec{U})
+
\varphi_{\mathrm{full}_3}(A,\vec{U})
\Big]
\xrightarrow{p}
\psi_a .
\]

Combining these results, it follows that
\begin{align*}
\hat \psi_a^{(S3, IF)}
&=
\mathbb{E}_n\!\Big[
\mathbb{E}_{n,\mathrm{sub}}\!\Big[
\sum_i
\mathcal{C}_{WZ}(A,\vec{u}_i)\,
\varphi_{\mathrm{full}_1}(Y,\vec{M},A,\vec{u}_i)
\Big]
\ \Big|\ Y,\vec{M},A,\vec{U}
\Big] \\
&\quad+
\mathbb{E}_n\!\Big[
\mathbb{E}_{n,\mathrm{sub}}\!\Big[
\sum_i
\Big\{
\mathcal{C}_{WZ}(A,\vec{u}_i)
+
\mathcal{C}_{YW}(\vec{M},A,\vec{u}_i)
+
\mathcal{C}_{YZ}(\vec{M},A,\vec{u}_i)
-2\,\mathcal{C}_{WZY}(\vec{M},A,\vec{u}_i)
\Big\}  \\
&\qquad\qquad\qquad\qquad\cdot
\Big\{
\varphi_{\mathrm{full}_2}(\vec{M},A,\vec{u}_i)
+
\varphi_{\mathrm{full}_3}(A,\vec{u}_i)
\Big\}
\Big]
\ \Big|\ \vec{M},A,\vec{U}
\Big]
\xrightarrow{p}
0 + \psi_a  = \psi_a.
\end{align*}

\textbf{Case 3.} \(p(\vec{M} \mid A, \vec{U})\), \(\{\mathbb{E}[f_j(Z) \mid \vec{U},A]\}_j\), \(\{\mathbb{E}[t_j(Y) \mid \vec{M},A,\vec{U}]\}_j\) are correctly specified.

The same reasoning as in Case~2 applies, replacing \(\{\mathbb{E}[h_j(W) \mid \vec{U}]\}_j\) with \(\{\mathbb{E}[f_j(Z) \mid \vec{U},A]\}_j\).

\textbf{Case 4.} \(\{\mathbb{E}[t_j(Y) \mid \vec{M},A,\vec{U}]\}_j\), \(\pi(A \mid \vec{U})\), \(\{\mathbb{E}[h_j(W) \mid \vec{U}]\}_j\) are correctly specified.

The same reasoning as Case 2 applies, noting \citet{fulcher2020robust} show that, provided \(\{\mathbb{E}[t_j(Y) \mid \vec{M},A,\vec{U}]\}_j\) and \(\pi(A \mid \vec{U})\) are correctly specified,
\[
\mathbb{E}_n\!\Big[
\varphi_{\mathrm{full}_2}(\vec{M},A,\vec{U})
+
\varphi_{\mathrm{full}_3}(A,\vec{U})
\Big]
\xrightarrow{p}
\psi_a .
\]

\textbf{Case 5.} \(\{\mathbb{E}[t_j(Y) \mid \vec{M},A,\vec{U}]\}_j\), \(\pi(A \mid \vec{U})\), \(\{\mathbb{E}[f_j(Z) \mid \vec{U},A]\}_j\) are correctly specified.

The same reasoning as in Case~4 applies, replacing \(\{\mathbb{E}[h_j(W) \mid \vec{U}]\}_j\) with \(\{\mathbb{E}[f_j(Z) \mid \vec{U},A]\}_j\).

\end{proof}

\section{Proof of Theorem~\ref{thm:asymp_normality}}\label{app:asymp_normality}
\begin{proof}
Let $\varphi = \varphi(O; \psi_a, \eta)$ from Theorem~\ref{thm:obs_if}. Expand \[
\hat{\psi_a} - \psi_a 
= \mathbb{P}_n \hat{\varphi} - \mathbb{P}\varphi
= (\mathbb{P}_n - \mathbb{P})(\hat{\varphi} - \varphi) + \mathbb{P}(\hat{\varphi} - \varphi) + (\mathbb{P}_n - \mathbb{P})\varphi
= R_1 + R_2 + (\mathbb{P}_n - \mathbb{P})\varphi,
\]
where 
\(
R_1 = (\mathbb{P}_n - \mathbb{P})(\hat{\varphi} - \varphi), 
\quad 
R_2 = \mathbb{P}(\hat{\varphi} - \varphi).
\)

$R_1 = o_p(n^{-1/2})$ under consistency of nuisance functions and sample splitting by Lemma 2 in \citet{kennedy2020sharp}.

Below we show $R_2 = o_p(n^{-1/2})$ under consistency of nuisance functions and conditions (1)-(5), therefore
\[
\hat{\psi_a} - \psi_a
= (\mathbb{P}_n - \mathbb{P})\varphi + o_p(n^{-1/2})
= \mathbb{P}_n\varphi + o_p(n^{-1/2})
\]
which is equivalent to \(
\sqrt{n}\,(\hat{\psi_a}-\psi_a) \;\;\overset{d}{\longrightarrow}\;\; 
\mathcal{N}\!\big(0, \,\mathbb{E}[\varphi^2]\big).
\)

Note that 
\begin{align*}
&R_2 =
\mathbb{P}\!\left[\hat{\varphi}-\varphi\right]\\
& \quad =
\sum_i\mathbb{E}\Big[
\hat{\mathcal C}_{WZ}(\vec{A,u_i})\,
\hat{\varphi}_{\mathrm{full},1}(Y,\vec{M},A,\vec{u_i})
\big] \\
& \quad
+\sum_i\mathbb{E}\Big[
\{\hat{\mathcal C}_{WZ}(\vec{A,u_i}) + \hat{\mathcal C}_{YW}(\vec{M},A,\vec{u_i}) + \hat{\mathcal C}_{YZ}(\vec{M},A,\vec{u_i}) -2 \hat{\mathcal C}_{YWZ}(\vec{M},A,\vec{u_i})\}\{\hat{\varphi}_{\mathrm{full},2}(\vec{M},A,\vec{u_i})
+
\hat{\varphi}_{\mathrm{full},3}(A,\vec{u_i})\}
\Big] -\psi_a.
\end{align*}

Consider the first term 
\begin{align*}
&\sum_i\mathbb{E}\Big[
\hat{\mathcal C}_{WZ}(A.\vec{u_i})\,
\hat{\varphi}_{\mathrm{full},1}(Y,\vec{M},A,\vec{u_i})
\big] \\
&= \sum_i\mathbb{E}\Big\{\frac{\hat p(\vec{M} \mid A=a, \vec{u_i})}{\hat p(\vec{M} \mid A, \vec{u_i})}\mathbb{E}[\mathcal{\hat C}_{WZ}(A,\vec{u_i})\{Y-\hat{\mu}(\vec{M},A,\vec{u_i})\}\mid \vec{M},A,\vec{U}]\Big\} \\
&= \sum_i\mathbb{E} \Big\{\frac{\hat p(\vec{M} \mid A=a, \vec{u_i})}{\hat p(\vec{M} \mid A, \vec{u_i})}\mathbb{E}[\mathcal{\hat C}_{WZ}(A,\vec{u_i})\mid \vec{M},A,\vec{U}][\mu(\vec{M},A,\vec{U})-\hat{\mu}(\vec{M}A,\vec{u_i})]\Big\}.
\end{align*}

Recall the following from Appendix~\ref{app:proxy_contrasts}:
\begin{align*}
&\mathbb{E}[\mathcal{\hat C}_{WZ}(\vec{A,u_i}) \mid \vec{M},A,\vec{U} = \vec{u_j}] \\
&= 
\prod_{j \neq i} \Big\{
\frac{\mathbb{E}\!\left[h_j(W)\mid \vec{U}=\vec{u_j}\right] - \hat{\mathbb{E}}\!\left[h_j(W)\mid \vec{U}=\vec{u_j}\right]}
     {\hat{\mathbb{E}}\!\left[h_j(W)\mid \vec{U}=\vec{u_i}\right]
      - \hat{\mathbb{E}}\!\left[h_j(W)\mid \vec{U}=\vec{u_j}\right]} \cdot
\frac{\mathbb{E}\!\left[g_j(Z)\mid \vec{U}=\vec{u_j},A\right]- \hat{\mathbb{E}}\!\left[g_j(Z)\mid \vec{U}=\vec{u_j},A\right]}
     {\hat{\mathbb{E}}\!\left[g_j(Z)\mid \vec{U}=\vec{u_i},A\right]
      - \hat{\mathbb{E}}\!\left[g_j(Z)\mid \vec{U}=\vec{u_j},A\right]} \Big\}.
\end{align*}

Then by (5), $\mathbb{E}[\mathcal{\hat C}_{WZ}(\vec{u_j})\mid \vec{M},A,\vec{u_j}] = o_p(n^{-1/2})$.

Additionally we have that
\begin{align*}
&\mathbb{E}[\mathcal{\hat C}_{WZ}(\vec{A,u_i}) \mid \vec{M},A,\vec{U} = \vec{u_i}] \\
&= 
\prod_{j \neq i} \Big\{
\frac{\mathbb{E}\!\left[h_j(W)\mid \vec{U}=\vec{u_i}\right] - \hat{\mathbb{E}}\!\left[h_j(W)\mid \vec{U}=\vec{u_j}\right]}
     {\hat{\mathbb{E}}\!\left[h_j(W)\mid \vec{U}=\vec{u_i}\right]
      - \hat{\mathbb{E}}\!\left[h_j(W)\mid \vec{U}=\vec{u_j}\right]} \cdot
\frac{\mathbb{E}\!\left[g_j(Z)\mid \vec{U}=\vec{u_i},A\right]- \hat{\mathbb{E}}\!\left[g_j(Z)\mid \vec{U}=\vec{u_j},A\right]}
     {\hat{\mathbb{E}}\!\left[g_j(Z)\mid \vec{U}=\vec{u_i},A\right]
      - \hat{\mathbb{E}}\!\left[g_j(Z)\mid \vec{U}=\vec{u_j},A\right]} \Big\} \\
&= 
\prod_{j \neq i} \Big\{
\frac{\mathbb{E}\!\left[h_j(W)\mid \vec{U}=\vec{u_i}\right]  - \mathbb{\hat E}[h_j(W)\mid \vec{U}=\vec{u_i}] + \mathbb{\hat E}[h_j(W)\mid \vec{U}=\vec{u_i}] - \hat{\mathbb{E}}\!\left[h_j(W)\mid \vec{U}=\vec{u_j}\right]}
     {\hat{\mathbb{E}}\!\left[h_j(W)\mid \vec{U}=\vec{u_i}\right]
      - \hat{\mathbb{E}}\!\left[h_j(W)\mid \vec{U}=\vec{u_j}\right]} \\
&\quad \cdot
\frac{\mathbb{E}\!\left[g_j(Z)\mid \vec{U}=\vec{u_i},A\right] - \mathbb{\hat E}\!\left[g_j(Z)\mid \vec{U}=\vec{u_i},A\right] + \mathbb{\hat E}\!\left[g_j(Z)\mid \vec{U}=\vec{u_i},A\right] - \hat{\mathbb{E}}\!\left[g_j(Z)\mid \vec{U}=\vec{u_j},A\right]}
     {\hat{\mathbb{E}}\!\left[g_j(Z)\mid \vec{U}=\vec{u_i},A\right]
      - \hat{\mathbb{E}}\!\left[g_j(Z)\mid \vec{U}=\vec{u_j},A\right]} \Big\} \\
&= 
\prod_{j \neq i} \Big\{
\frac{\mathbb{E}\!\left[h_j(W)\mid \vec{U}=\vec{u_i}\right]  - \mathbb{\hat E}[h_j(W)\mid \vec{U}=\vec{u_i}]}
     {\hat{\mathbb{E}}\!\left[h_j(W)\mid \vec{U}=\vec{u_i}\right]
      - \hat{\mathbb{E}}\!\left[h_j(W)\mid \vec{U}=\vec{u_j}\right]} +1 \Big\}\cdot
\Big\{
\frac{\mathbb{E}\!\left[g_j(Z)\mid \vec{U}=\vec{u_i},A\right] - \mathbb{\hat E}\!\left[g_j(Z)\mid \vec{U}=\vec{u_i},A\right]}
     {\hat{\mathbb{E}}\!\left[g_j(Z)\mid \vec{U}=\vec{u_i},A\right]
      - \hat{\mathbb{E}}\!\left[g_j(Z)\mid \vec{U}=\vec{u_j},A\right]} +1 \Big\}.
\end{align*}

By (3) and (4), $\mathbb{E}[\mathcal{\hat C}_{WZ}(\vec{A,u_i})\mid \vec{M},A,\vec{u_i}][\mu(\vec{M},A,\vec{u_i})-\hat{\mu}(\vec{M}A,\vec{u_i})] = o_p(n^{-1/2})$. Hence, the first term is \(\sqrt{n}\)-consistent, i.e., \(\sum_i\mathbb{E}\Big[
\hat{\mathcal C}_{WZ}(\vec{A,u_i})\,
\hat{\varphi}_{\mathrm{full},1}(Y,\vec{M},A,\vec{u_i})
\big] = o_p(n^{-1/2})\).

Now consider the second term 
\begin{align*}
&\sum_i\mathbb{E}\Big[
\{\hat{\mathcal C}_{WZ}(\vec{A,u_i}) + \hat{\mathcal C}_{YW}(\vec{M},A,\vec{u_i}) + \hat{\mathcal C}_{YZ}(\vec{M},A,\vec{u_i}) -2 \hat{\mathcal C}_{YWZ}(\vec{M},A,\vec{u_i})\}\{\hat{\varphi}_{\mathrm{full},2}(\vec{M},A,\vec{u_i})
+
\hat{\varphi}_{\mathrm{full},3}(A,\vec{u_i})\} \Big] -\psi_a \\
&= \sum_i\mathbb{E}\Big[\{\hat{\varphi}_{\mathrm{full},2}(\vec{M},A,\vec{u_i}) 
+
\hat{\varphi}_{\mathrm{full},3}(A,\vec{u_i})\}
\{\mathbb{E}[\hat{\mathcal C}_{WZ}(A,\vec{u_i}) \mid \vec{M},A,\vec{U}] + \mathbb{E}[\hat{\mathcal C}_{YW}(\vec{M},A,\vec{u_i}) \mid \vec{M},A,\vec{U}] \\
&\quad + \mathbb{E}[\hat{\mathcal C}_{YZ}(\vec{M},A,\vec{u_i})\mid \vec{M},A,\vec{U}] -2 \mathbb{E}[\hat{\mathcal C}_{YWZ}(\vec{M},A,\vec{u_i})\mid \vec{M},A,\vec{U}]\} \Big] -\psi_a.
\end{align*}

By (3)-(5), $\mathbb{E}[\mathcal{\hat C}_{WZ}(A,\vec{u_j})\mid \vec{M},A,\vec{u_j}] = o_p(n^{-1/2})$, $\mathbb{E}[\mathcal{\hat C}_{YW}(\vec{M},A,\vec{u_j})\mid \vec{M},A,\vec{u_j}] = o_p(n^{-1/2})$, $\mathbb{E}[\mathcal{\hat C}_{YZ}(\vec{M},A,\vec{u_j})\mid \vec{M},A,\vec{u_j}]= o_p(n^{-1/2})$, and $\mathbb{E}[\mathcal{\hat C}_{YWZ}(\vec{M},A,\vec{u_j})\mid \vec{M},A,\vec{u_j}] = o_p(n^{-1/2})$.

Additionally, we have that
\begin{align*}
&\mathbb{E}[\mathcal{\hat C}_{WZ}(A,\vec{u_i}) \mid \vec{M},A,\vec{U} = \vec{u_i}] \\
&= 
\prod_{j \neq i} \Big\{
\frac{\mathbb{E}\!\left[h_j(W)\mid \vec{U}=\vec{u_i}\right]  - \mathbb{\hat E}[h_j(W)\mid \vec{U}=\vec{u_i}]}
     {\hat{\mathbb{E}}\!\left[h_j(W)\mid \vec{U}=\vec{u_i}\right]
      - \hat{\mathbb{E}}\!\left[h_j(W)\mid \vec{U}=\vec{u_j}\right]} +1 \Big\}\cdot
\Big\{
\frac{\mathbb{E}\!\left[g_j(Z)\mid \vec{U}=\vec{u_i},A\right] - \mathbb{\hat E}\!\left[g_j(Z)\mid \vec{U}=\vec{u_i},A\right]}
     {\hat{\mathbb{E}}\!\left[g_j(Z)\mid \vec{U}=\vec{u_i},A\right]
      - \hat{\mathbb{E}}\!\left[g_j(Z)\mid \vec{U}=\vec{u_j},A\right]} +1 \Big\},
\end{align*}
and likewise for \(\mathbb{E}[\mathcal{\hat C}_{YW}(\vec{M}, A, \vec{u_i}) \mid \vec{M},A,\vec{U} = \vec{u_i}]\) and \(\mathbb{E}[\mathcal{\hat C}_{YZ}(\vec{M},A,\vec{u_i}) \mid \vec{M},A,\vec{U} = \vec{u_i}]\). 

Similarly,
\begin{align*}
&\mathbb{E}[\mathcal{\hat C}_{YWZ}(\vec{M},A,\vec{u_i}) \mid \vec{M},A,\vec{U} = \vec{u_i}] \\
&= 
\prod_{j \neq i} 
\Big\{
\frac{\mathbb{E}\!\left[t_j(Y)\mid \vec{M},A,\vec{U}=\vec{u_i}\right] - \mathbb{\hat E}\!\left[t_j(Y)\mid \vec{M},A,\vec{U}=\vec{u_i}\right]}
     {\hat{\mathbb{E}}\!\left[t_j(Y)\mid \vec{M},A,\vec{U}=\vec{u_i}\right]
      - \hat{\mathbb{E}}\!\left[t_j(Y)\mid \vec{M},A,\vec{U}=\vec{u_j}\right]} +1 \Big\} \cdot
      \Big\{
\frac{\mathbb{E}\!\left[h_j(W)\mid \vec{U}=\vec{u_i}\right]  - \mathbb{\hat E}[h_j(W)\mid \vec{U}=\vec{u_i}]}
     {\hat{\mathbb{E}}\!\left[h_j(W)\mid \vec{U}=\vec{u_i}\right]
      - \hat{\mathbb{E}}\!\left[h_j(W)\mid \vec{U}=\vec{u_j}\right]} +1 \Big\} \\
&\quad \cdot
\Big\{
\frac{\mathbb{E}\!\left[g_j(Z)\mid \vec{U}=\vec{u_i},A\right] - \mathbb{\hat E}\!\left[g_j(Z)\mid \vec{U}=\vec{u_i},A\right]}
     {\hat{\mathbb{E}}\!\left[g_j(Z)\mid \vec{U}=\vec{u_i},A\right]
      - \hat{\mathbb{E}}\!\left[g_j(Z)\mid \vec{U}=\vec{u_j},A\right]} +1 \Big\}.
\end{align*}

By (3)--(5),
\[
\begin{aligned}
&\mathbb{E}[\hat{\mathcal C}_{WZ}(A,\vec u_i) \mid \vec M,A,\vec u_i]
+\mathbb{E}[\hat{\mathcal C}_{YW}(\vec{M},A,\vec u_i) \mid \vec M,A,\vec u_i] \\
&\quad
+\mathbb{E}[\hat{\mathcal C}_{YZ}(\vec{M},A,\vec u_i)\mid \vec M,A,\vec u_i]
-2\mathbb{E}[\hat{\mathcal C}_{YWZ}(\vec{M},A,\vec u_i)\mid \vec M,A,\vec u_i]
= 1 + o_p(n^{-1/2}).
\end{aligned}
\]

Thus, we have for the second term 
\begin{align*}
&\sum_i\mathbb{E}\Big[
\{\hat{\mathcal C}_{WZ}(A,\vec{u_i}) + \hat{\mathcal C}_{YW}(\vec{M},A,\vec{u_i}) + \hat{\mathcal C}_{YZ}(\vec{M},A,\vec{u_i}) -2 \hat{\mathcal C}_{YWZ}(\vec{M},A,\vec{u_i})\}\{\hat{\varphi}_{\mathrm{full},2}(\vec{M},A,\vec{u_i})
+
\hat{\varphi}_{\mathrm{full},3}(A,\vec{u_i})\} \Big] -\psi_a \\
&= \sum_i\mathbb{E}\Big[\{\hat{\varphi}_{\mathrm{full},2}(\vec{M},A,\vec{u_i}) 
+
\hat{\varphi}_{\mathrm{full},3}(A,\vec{u_i})\}
\{\mathbb{E}[\hat{\mathcal C}_{WZ}(A,\vec{u_i}) \mid \vec{M},A,\vec{U}] + \mathbb{E}[\hat{\mathcal C}_{YW}(\vec{M},A,\vec{u_i}) \mid \vec{M},A,\vec{U}] \\
&\quad + \mathbb{E}[\hat{\mathcal C}_{YZ}(\vec{M},A,\vec{u_i})\mid \vec{M},A,\vec{U}] -2 \mathbb{E}[\hat{\mathcal C}_{YWZ}(\vec{M},A,\vec{u_i})\mid \vec{M},A,\vec{U}]\} \Big] -\psi_a  \\
&= \sum_i\mathbb{E}\Big[\{\hat{\varphi}_{\mathrm{full},2}(\vec{M},A,\vec{u_i}) 
+
\hat{\varphi}_{\mathrm{full},3}(A,\vec{u_i})\}
\mathbf{1}\{\vec U = \vec u_i\} \Big] -\psi_a + o_p(n^{-1/2}) \\
&= \mathbb{E}\Big[\hat{\varphi}_{\mathrm{full},2}(\vec{M},A,\vec{U}) 
+
\hat{\varphi}_{\mathrm{full},3}(A,\vec{U})
\Big] -\psi_a + o_p(n^{-1/2}).
\end{align*}

Theorem 5.1 from \citet{guo2026flexiblefrontdoor} showed that under (1) and (2), \(\mathbb{E}\Big[\hat{\varphi}_{\mathrm{full},2}(\vec{M},A,\vec{U}) 
+
\hat{\varphi}_{\mathrm{full},3}(A,\vec{U})
\Big] -\psi_a= o_p(n^{-1/2})\). Hence, the second term is also \(\sqrt{n}\)-consistent.

Altogether, \(R_2 = o_p(n^{-1/2})\) as desired.
\end{proof}

\section{Experiments: Finite-Support DGP Data Generation }\label{app:DGP}

The setup is based on the causal model reflected in Fig.\ref{fig:primal_v3}:
\begin{align}
& U \sim Ber(0.1) \\
& H \sim Ber(0.8) \\
& A \mid UH \sim Ber(0.4 + 0.2 * U +0.1*H) \\ 
& M \mid AU \sim Ber(0.8 - 0.5*xor(U,A)+0.1*A) \\ 
& Y \mid MUH \sim Ber(max(min(\\& \quad 0.5-0.4*U+0.6*M-0.1*H,0.9),0.1)) \\
& Z \mid U \sim Ber(0.3+0.5*U) \\ 
& W \mid U \sim Ber(0.9-0.7*U) 
\end{align}

\section{Experiments: Mixed DGP Data Generation }\label{app:DGP2}

The setup is based on the causal model reflected in Fig.~\ref{fig:primal_v3}:
\begin{align}
& U_1 \sim Ber(0.4) \\
& U_2 \sim N(0,1) \\
& A \mid U_1,U_2 \sim Ber\!\left(\expit(-0.5 + 0.8U_1 + 0.6U_2)\right) \\
& M \mid A,U_1 \sim Ber\!\left(\expit(-0.3 + 1.3U_1 + 0.9A)\right) \\
& Y \mid M,U_1,U_2 \sim N(1 + 2U_1 + 0.8U_2 + 0.9M,\,1) \\
& W \mid U_1 \sim N(0.4 + U_1,\,1) \\
& Z \mid U_1 \sim N(-0.2 + 1.1U_1,\,1).
\end{align}

All exogenous noise terms are mutually independent.

\section{The benchmark estimators } \label{app:bm_estimators}
We used the following benchmark estimators:  

An oracle plug-in estimator \(\hat{\psi}^{(\mathrm{Oracle,PI})}\) (Equation~\ref{eq:S1_Oracle}) that has access to the full data, had $U$ not been hidden: 

\begin{equation}\label{eq:S1_Oracle}
\hat{\psi}^{(\mathrm{Oracle},PI)} 
= \mathbb{E}_n \Big[
\sum_{a'} 
\hat E[Y \mid M,a',U] \,
\hat p(a' \mid U) \,
[\hat p(M \mid a=1,U)-\hat p(M \mid a=0,U)]
\Big]
\end{equation}

A front-door plug-in estimator \(\hat{\psi}^{(\mathrm{FD,PI})}\) (Equation \ref{eq:S3_FD}) that falsely assumes \(\vec{M}\) satisfies the front-door criterion. 

\begin{equation}\label{eq:S3_FD} 
\hat{\psi}^{(\mathrm{FD},PI)} 
= \mathbb{E}_n \Big[
\sum_{a'} 
\hat E[Y \mid M,a'] \,
\hat p(a') \,
[\hat p(M \mid a=1)-\hat p(M \mid a=0)]
\Big]
\end{equation}